%% file: maic_arxiv.tex
\documentclass{article}
\usepackage{graphicx} % Required for inserting images
\usepackage{hyperref}
\usepackage{xcolor} % Include the package at the beginning
\usepackage{fullpage}
\usepackage{algorithm}
\usepackage{amsmath,amssymb,amsthm}
\usepackage{dsfont}
\usepackage{multirow}
\usepackage{url}
\usepackage{cancel}   % for the cancel symbol
\usepackage[normalem]{ulem} % Include the package at the beginning. The [normalem] option keeps \emph{} as italics rather than underline.
\usepackage[square,sort&compress,numbers]{natbib}
\usepackage{geometry}
\input{macros}

\newtheorem*{claim}{Claim}

\usepackage{epsfig}
\usepackage{multienum}
\usepackage[noend]{algpseudocode}
\usepackage{ifsym}
\usepackage{verbatim}
\usepackage{color}
\usepackage{bbm}
\usepackage{hyperref}
\usepackage{todonotes}
\usepackage{subcaption}
\usepackage{DS-style}
\usepackage{tikz}
\usetikzlibrary{shadings}

\makeatletter
\newcommand{\headline}[2]{%
  \vspace{0.08in}%
  \noindent \textbf{#1} #2%
  \vspace{0.08in}%
  \par\nobreak
  \@afterindentfalse
  \@afterheading
}
\makeatother

\title{Interdependent Bilateral Trade: Information vs Approximation \thanks{The work of S. Dobzinski and A. Shaulker was supported by the BSF-NSF grant (BSF number: 2021655, NSF number: 2127781).
The work of A. Eden was supported by the Israel Science Foundation (grant No. 533/23). Alon Eden is the Incumbent of the Harry $\&$ Abe Sherman Senior Lectureship at the School of Computer Science and Engineering at the Hebrew University.
The work of K. Goldner and T. Tsilivis was supported by NSF CAREER Award CCF-2441071. }}

\begin{document}

% Title page for title and abstract only.
% \begin{titlepage}

\author{Shahar Dobzinski \thanks{Weizmann Institute of Science (shahar.dobzinski@weizmann.ac.il)} \and
 Alon Eden \thanks{Hebrew University (alon.eden@mail.huji.ac.il)} \and Kira Goldner \thanks{Boston University (goldner@bu.edu)} \and Ariel Shaulker \thanks{Weizmann Institute of Science (ariel.shaulker@weizmann.ac.il)} \and Thodoris Tsilivis  \thanks{Boston University (tsilivis@bu.edu)}}

% \author{Shahar Dobzinski \and Alon Eden \and Kira Goldner \and Ariel Shaulker \and Thodoris Tsilivis}
% \date{January 2025}

% \begin{acks}
% The work of S. Dobzinski and A. Shaulker was supported by the BSF-NSF grant (BSF number: 2021655, NSF number: 2127781).
% The work of A. Eden was supported by the Israel Science Foundation (grant No. 533/23). Alon Eden is the Incumbent of the Harry $\&$ Abe Sherman Senior Lectureship at the School of Computer Science and Engineering at the Hebrew University.
% The work of K. Goldner and T. Tsilivis was supported by NSF CAREER Award CCF-2441071. 
% \end{acks}

\maketitle

\begin{abstract}
Welfare maximization in bilateral trade has been extensively studied in recent years. Previous literature obtained incentive-compatible approximation mechanisms only for the private values case. In this paper, we study welfare maximization in bilateral trade with interdependent values. Designing mechanisms for interdependent settings is much more challenging because the values of the players depend on the private information of the others, requiring complex belief updates and strategic inference.

%We adopt two complementary approaches to analyze bilateral trade in the interdependent setting. The first and main approach classifies information structures based on the amount of influence that a player's received signal \kgc{private information?} has on their valuation. The second approach identifies meaningful families of information structures. For both approaches, we establish both possibility and impossibility results.

We propose to classify information structures by quantifying the influence that a player's private signal has on their own valuation. We then paint a picture of where approximations are possible and impossible based on these information structures.  Finally, we also study the possible approximation ratios for a natural family of information structures.
\end{abstract}

% Optionally include a table of contents
% \vspace{1cm}
% \setcounter{tocdepth}{2} % adjust to 1 if desired
% \tableofcontents

% \end{titlepage}

% PAPER BODY STARTS

\input{Introduction}

\input{Preliminaries}
\input{boundaries}
\input{polynomials}

\bibliographystyle{ACM-Reference-Format} %{plainnat}
\bibliography{refs}

\appendix
\input{obstacles}

\input{appproofs}

\end{document}

%% file: macros.tex
\newtheorem{lemma}{Lemma}
\newtheorem{theorem}{Theorem}
\newtheorem{corollary}[theorem]{Corollary}

\newtheorem{proposition}[theorem]{Proposition}

\theoremstyle{definition}
\newtheorem{definition}{Definition}

\newtheorem{example}{Example}
\makeatletter
\def\blfootnote{\xdef\@thefnmark{}\@footnotetext}
\makeatother

\newcommand{\prob}[2][]{\text{\bf Pr}\ifthenelse{\not\equal{}{#1}}{_{#1}}{}\!\left[#2\right]}
\newcommand{\expect}[2][]{\text{\bf E}\ifthenelse{\not\equal{}{#1}}{_{#1}}{}\!\left[#2\right]}

\newcommand{\E}{\mathbb{E}}

\def\Pr{\ensuremath{\mathrm{Pr}}}

\newcommand{\notshow}[1]{{}}

\newcommand*{\pr}[2][]{\text{Pr}\ifx\\\left[#1\right]\\\else_{#1}\fi \left[#2\right]}

\defcitealias{CDW}{CDW '16}
\defcitealias{FGKK}{FGKK '16}
\defcitealias{DW}{DW '17}
\defcitealias{DHP}{DHP '17}
\defcitealias{Myerson}{Mye '81}
\defcitealias{DDT15}{DDT '15}
\defcitealias{SSW17}{SSW '18}
\defcitealias{MV}{MV '07}
\defcitealias{CheGale2000}{Che and Gale 2000}
\defcitealias{DGSSW}{DGSSW '19}
\defcitealias{CDGM19}{CDGM '19}
\defcitealias{EFFGK}{EFFGK '19}

\newenvironment{numbredtheorem}[1]{%
\renewcommand{\thetheorem}{#1}%
\begin{theorem}}{\end{theorem}\addtocounter{theorem}{-1}}

%% file: Introduction.tex
\section{Introduction}

%%WHAT IS OUR SETTING: Welfare max in bilateral trade, and specifically, we'll talk about information which will explore the hierachy between indep/corr/interdep values.

%\kgc{Rewrote the next two paras a bit mostly for wording. All previous writing is still here commented.}

We investigate welfare maximization in bilateral trade.  Bilateral trade is the setting with one item, one buyer with value $v_b$, and one seller with value $v_s$. We aim to design mechanisms that satisfy budget-balance and incentive-compatibility (IC) constraints, and that (approximately) maximize expected social welfare. More specifically, in bilateral trade, the aim is to maximize the expected value of the player who retains the item at the end of the process.

%This paper studies the bilateral trade problem. In the bilateral trade problem, we have one item, a buyer with value $v_b$, and a seller with value $v_s$. The objective is to design budget-balanced, incentive-compatible trade mechanisms that (approximately) maximize social welfare—specifically, mechanisms that aim to maximize the expected value of the player who retains the item at the end of the process.

We explore the extent to which the quality of mechanisms' approximations is tied with the information that the players hold. Without prior information regarding the value of the other player, it is not hard to see that no mechanism can always guarantee a reasonable approximation to the welfare. Therefore, %{richer information models}
more information is necessary to achieve good approximations. In particular, consider the following hierarchy of canonical information models:

%It is not hard to see that the quality of the mechanisms is tied with the information that the players hold. If each player has no information at all on the value of the other player, it is not hard to see that no mechanism can always guarantee a reasonable approximation to the welfare. Therefore, richer information models are necessary to achieve good approximations. In particular, consider the following hierarchy of canonical information models:

\begin{enumerate}
\item \textbf{Independent private values.} The private value of each player is drawn from a known distribution, independently of the value of the other players.

\item \textbf{Correlated private values.} The private values of the players are drawn from a joint distribution.

\item \textbf{Interdependent private values.} Each player receives a private signal, and their value for the item is a public function of all signals. Thus, a player may not even know their own value, that is, they are not fully informed.
%be fully aware to its value.
\end{enumerate}
%We emphasize 
Note that this list is not exhaustive, as the models can be further distinguished and refined. For example, interdependent values contain the special case of the common value model, where all players share the same underlying valuation but receive different signals about it. However, the common value model is generally less relevant in our context since the mechanism that never trades the item maximizes social welfare. 

%%What's known for welf/bilateral and our contribution of information.
We can now contextualize welfare maximization in bilateral trade within this hierarchy. For independent private values, the seminal result of \citet{MS83} dictates that for bilateral trade there exist no IC (Bayesian or Dominant-Strategy IC), individually rational mechanism that simultaneously maximizes social welfare and maintains budget balance. Mechanisms satisfying all three properties with constant approximation guarantees to the optimal social welfare have been proposed in \cite{BD14, CKLT16, BD21, KPV21, LRW23, CW23}. For the problem of bilateral trade with correlated values, \citet{DS24} provide a constant-factor approximation to the optimal social welfare.

%For interdependent bilateral trade, \cite{KZ20} identifies conditions (SC, no-trade-no-payments condition) under which an optimal (pointwise social welfare (not expected)) two-stage mechanism that satisfies BIC, IIR and BB exists. For interdependent bilateral trade, \cite{KZ22} consider 2 (generalized two-stage Groves mechanism and shoot-the-liar mechanism) two-stage mechanisms and prove that in a setting with a fully informed seller and partially-informed buyer these mechanisms achieve BIC,IIR,BB and optimal (pointwise again, not expected) social welfare, while in a setting with a partially-informed seller and a partially-informed buyer these mechanisms may not satisfy BIC or IIR respectively.

Mechanisms for interdependent valuations were studied in both economics and computer science, primarily in the setting of auctions \cite{MilgromWeber82,JM01,DM00,RT-C16,EFFG18,EFFGK19,eden2021poa,EGZ22,EFFGKjournal,EFGMM23,eden2024private,chen2021cursed,CFK14}. In the context of bilateral trade, \cite{gresik, interBT,KZ22} examined conditions under which an optimal solution exists.
Brooks and Du \cite{brooksDu} analyze a setting with the additional assumption that, in every realization, the buyer's value exceeds the seller's by at least some known amount. They identify a condition under which there exist information structures where the gains-from-trade can be inapproximable.
In contrast, we study a much more general environment, and ask when approximation is even possible.
Settings with interdependent values are generally more challenging to analyze. The complexity arises because players have to infer their valuations based on others' private information, which may lead to complex belief updating and strategic inference~\citep{Akerlof70,MilgromWeber82}. To explain this, consider any posted price mechanism. In the independent private values model, incentive compatibility is immediate---each player's optimal strategy is simply to accept or reject the posted price based on their private valuation. The primary challenge lies, therefore, in proving that the mechanism achieves a good approximation ratio for welfare.

However, when values are interdependent, proving incentive compatibility becomes significantly more complex. In fact, an equilibrium does not necessarily exist. To see why, consider a price set by the mechanism designer. Trade occurs only if both players accept the price. Take, for instance, the seller's perspective. If the buyer rejects the offer, no trade will happen regardless of the seller's decision. Therefore, the seller should evaluate the offer assuming that the buyer accepts it, leading them to update their belief about their own value for the item. Similarly, the buyer must update their own value assuming that the seller accepts the trade. This process of mutual updating of beliefs might lead to a situation where no equilibrium exists. We discuss this with a concrete example in Appendix \ref{subsec-obstacle-posted-price}. 

%In this paper, we introduce the idea of the degree of informedness of each player

%In this paper, we propose to quantify the influence that a player's private signal has on their own valuation, and paint of picture of where approximation results are achieved or impossible as determined by player informedness.  We then study the power of approximation mechanisms on restricted families of information structures.

\subsection*{Our Model}

We consider a model of bilateral trade where the players' valuation functions are additively separable in both signals. That is, if $\signal{s}$ and $\signal{b}$ are the private signals received by the seller and the buyer, respectively, then there exist functions $\buyervaluationbuyerfunc$ and $\buyervaluationsellerfunc$ such that the buyer's value can be expressed as:
    \begin{equation*}
    % \label{equ:def-buyer-value}
        v_b(\signal{b}, \signal{s}) = \buyervaluationbuyerfunc{}(\signal{b}) + \buyervaluationsellerfunc(\signal{s})
    \end{equation*}
    and functions $\sellervaluationbuyerfunc$ and $\sellervaluationsellerfunc$ such that the seller's value can be expressed as:
    \begin{equation*}
    % \label{equ:def-seller-value}
        v_s(\signal{b}, \signal{s}) = \sellervaluationbuyerfunc{}(\signal{b}) + \sellervaluationsellerfunc(\signal{s}).
    \end{equation*}
We assume that $\buyervaluationbuyerfunc, \buyervaluationsellerfunc, \sellervaluationbuyerfunc,\sellervaluationsellerfunc$ are non-decreasing\footnote{We note that all of our positive results holds without this assumption.} and that $\signal{s},\signal{b}$ are independently drawn from the interval $[0,1]$. Each tuple of $\buyervaluationbuyerfunc, \buyervaluationsellerfunc, \sellervaluationbuyerfunc,\sellervaluationsellerfunc$ defines an \emph{information structure}, as it determines the influence of the buyer and seller's signals on both the self and the other. As we work in the combination of the interdependent values and two-sided settings, we seek mechanisms that are Bayesian incentive-compatible and interim individually rational.\footnote{Indeed, in Appendix \ref{subsec-obstacle-ex-post} we demonstrate that stronger notions like ex-post incentive compatibility and ex-post individual rationality have little power in our setting.} 

As we will see later, it is not hard to come up with information structures for which there is no Bayesian incentive-compatible mechanism that can achieve a good approximation. Thus, our goal is to obtain a better understanding of which information structures admit a Bayesian incentive-compatible, budget-balanced mechanism that achieves a good approximation to social welfare and which are not. 

In this paper, our technical contributions take two different approaches toward classifying information structures. Our main approach does not restrict $\buyervaluationbuyerfunc, \buyervaluationsellerfunc, \sellervaluationbuyerfunc,\sellervaluationsellerfunc$ at all and classifies information structures by the amount of influence of the signal that the player receives on their valuation. In addition, we also take another approach and identify some meaningful families of information structures where $\buyervaluationbuyerfunc, \buyervaluationsellerfunc, \sellervaluationbuyerfunc,\sellervaluationsellerfunc$ have a specific structure.

\subsection*{Main Approach: Amount of Information vs. Approximation}

In our first approach, we measure the influence of the signal that each player receives on their value. Towards this end, we say that a seller is $\informedcoeffseller$-informed if  $\informedcoeffseller = \frac{\mathbb{E}_{\signal{s}}[\sellervaluationsellerfunc(\signal{s})]}{\mathbb{E}_{\signal{s}}[\sellervaluationsellerfunc(\signal{s})]+ \mathbb{E}_{\signal{b}}[\sellervaluationbuyerfunc{}(\signal{b})]}$. Similarly, a buyer is $\beta$-informed if $\beta = \frac{\mathbb{E}_{\signal{b}}[\buyervaluationbuyerfunc{}(\signal{b})]}{\mathbb{E}_{\signal{b}}[\buyervaluationbuyerfunc{}(\signal{b})]+ \mathbb{E}_{\signal{s}}[\buyervaluationsellerfunc{}(\signal{s})]}$.\footnote{We note that the assumption that the functions are additively separable is used only to simplify the presentation. All our positive results also hold for the case when the definitions of $\alpha, \beta$ are generalized to $\alpha = \frac{\E_{\signal{s}}[v_s(0_b,\signal{s})]}{ \E_{\signal{b},\signal{s}}[v_s(\signal{b},\signal{s})]}
    \text{\ and\ } 
    \beta = \frac{\E_{\signal{b}}[v_b(\signal{b},0_s)]}{ \E_{\signal{b},\signal{s}}[v_b(\signal{b},\signal{s})]}$. Of course, our negative results continue to hold in this more general case.}

This approach allows us to classify every possible information structure by mapping each one to a pair of numbers $(\alpha,\beta)$. In other words, each possible information structure corresponds to a point in the unit square. See figure \ref{fig:intro-information-rectangle} for an illustration. Let us start with considering the four extreme points of the unit square: $(0,0), (0,1), (1,0)$, and $(1,1)$.

The extreme point that is the easiest to analyze is $(0,0)$. In this case, we say that both players are uninformed since the signal that each player receives has no effect on their value. Let $E_b=\mathbb{E}_{\signal{s}}[\buyervaluationsellerfunc(\signal{s})]$ denote the expected value of the buyer and $E_s=\mathbb{E}_{\signal{b}}[\sellervaluationbuyerfunc{}(\signal{b})]$ denote the expected value of the seller. When $E_s \geq E_b$, the algorithm that keeps the item with the seller already provides a $2$-approximation: the expected value of the optimal welfare is at most $E_s+E_b\leq 2\cdot E_s$ whereas the welfare of the algorithm is $E_s$. If $E_b > E_s$, then we set a price of $\frac {E_b+E_s} 2$ for trade. Both parties will always agree to trade regardless of their signals, and we get a $2$-approximation in this case as well.

The extreme point $(1,1)$ was already considered in the literature. Here, the two bidders are completely informed, so this point corresponds to the model of independent private values. As mentioned earlier, this model was extensively studied and a constant approximation can be achieved by a posted price \cite{BD14, BD21, KPV21, CKLT16, CW23, LRW23}.

The two remaining extreme points---a fully informed buyer and an uninformed seller, and a fully informed seller and an uninformed buyer---are more challenging to analyze. However, for both points we are able to prove that there are information structures for which no Bayesian incentive-compatible mechanism can guarantee any constant approximation. The $(0,1)$ impossibility holds even when the values are single crossing---that is, one can implement the welfare-maximizing mechanism in an incentive-compatible way, but the budget balance requirement prevents the existence of any incentive-compatible mechanism that maximizes welfare even approximately. In contrast, the $(1,0)$ impossibility does not hold when the values are single-crossing, for a good reason, as we will later see. 

We now move on to analyzing the ``edges'' between these extreme points. We start with an impossibility when the buyer is uninformed:

\begin{figure}[h]  % Position the figure here
    \centering  % Center the figure
\begin{tikzpicture}
    % seller is informed
    \shade[top color=red, bottom color=blue] (0,0) rectangle (0.1,6);
    % Buyer is informed
    % \shade[left color=red, right color=blue] (0.05,5) rectangle (6.05,6.05);
    \fill[red] (0.1,6) rectangle (6.1,6.1);
    % seller is uninformed
    \shade[top color=blue, bottom color=red] (6,0) rectangle (6.1,6);
    % buyer is uninformed
    % \shade[left color=blue, right color=red] (0,0) rectangle (5.05,0.05);
    \fill[red] (0,0) rectangle (6.1,0.1);
    % Label the vertices with original tags and color the vertices
    \node at (0.05,0.05) [fill=blue, text=white, circle, inner sep=2pt ] {};
    \node[scale=0.8] at (-0.6,-0.4) {\shortstack{$(0,0)$ \\uninformed \\ players}};

    \node at (6.05,0.05) [fill=red, text=white, circle, inner sep=2pt] {};
    \node[scale=0.8] at (6.7,-0.4) {\shortstack{$(1,0)$ \\ fully informed seller \\ uninformed buyer}};

    \node at (0.05,6.05) [fill=red, text=white, circle, inner sep=2pt] {};
    \node[scale=0.8] at (-0.6,6.7) {\shortstack{ uninformed seller \\ fully informed buyer\\ $(0,1)$ }};

    \node at (6.05,6.05) [fill=blue, text=white, circle, inner sep=2pt] {};
    \node[scale=0.8] at (6.7,6.6) {\shortstack{\\private values\\ $(1,1)$ }};

    \node at (3, -0.8) {Increasingly informed seller};
    \draw[->, thick] (1, -0.5) -- (5, -0.5); 
   
   \node[rotate=90] at (-0.8, 3.2) {Increasingly informed buyer};
   \draw[->, thick] (-0.5, 1.2) -- (-0.5, 5.2); 
   
\end{tikzpicture}
\caption{\small An illustration of our results. Each point \((\informedcoeffseller, \informedcoeffbuyer)\) corresponds to an \(\informedcoeffseller\)-informed seller and a \(\informedcoeffbuyer\)-informed buyer. If a point $(\alpha,\beta)$ is colored blue, there is a constant-factor mechanism for every $(\alpha,\beta)$-information structure. If a point $(\alpha,\beta)$ is colored red, there exists an $(\alpha,\beta)$-information structure for which no mechanism can guarantee a constant approximation. For gradient edges, the approximation ratio we guarantee smoothly varies as the amount of information the players hold changes. The right edge is discussed in Section~\ref{sec:seller-informed}, the top edge in Section~\ref{sec:buyer-informed}, the left edge in Section~\ref{sec:seller-uninformed}, and the bottom edge in Section~\ref{sec:uninformed-buyer}.
}
\label{fig:intro-information-rectangle}
\end{figure}

% \paragraph{\textbf{Informal Theorem: }}Fix $1 \geq \alpha >0$. There exists an $(\alpha,0)$-information structure where no Bayesian incentive-compatible mechanism provides a constant approximation ratio.

% \begingroup
% \raggedleft 
% \textbf{Informal Theorem: }Fix $1 \geq \alpha >0$. There exists an $(\alpha,0)$-information structure where no Bayesian incentive-compatible mechanism provides a constant approximation ratio.
% \endgroup

\headline{Informal Theorem: }{Fix $1 \geq \alpha >0$. There exists an $(\alpha,0)$-information structure where no Bayesian incentive-compatible mechanism provides a constant approximation ratio.}

The theorem, along with our discussion above, highlights an interesting ``threshold'' phenomenon: if both the seller and the buyer are fully uninformed, there is a mechanism that provides a constant approximation ratio. However, there exist information structures where the signal that the seller receives is only slightly informative, but no mechanism can provide a constant approximation ratio for these structures (Section~\ref{sec:uninformed-buyer}).

In contrast, if the seller is uninformed, then the possible approximation ratio worsens as the buyer becomes more informed. In other words, ignorance is bliss: less information helps us to obtain better approximation ratios (Section~\ref{sec:seller-uninformed}):

\headline{Informal Theorem: }{Fix $1 \geq \beta >0$. For every $(0,\beta)$-information structure, there exists a Bayesian incentive-compatible mechanism that provides an approximation ratio of $O(\frac 1 {1-\beta}).$ Furthermore, for any $\informedcoeffbuyer \geq 0.9$, there exists an $(0, \informedcoeffbuyer)$-information structure where no Bayesian incentive-compatible mechanism has an approximation ratio better than $\Omega(\sqrt{\frac{1}{1-\informedcoeffbuyer}})$.}

We now proceed to consider information structures where the seller is fully informed and the buyer is partially informed. We show that, given a mechanism that provides a $c$-approximation in the independent private values model, we can construct another mechanism that provides a good approximation ratio in our case (Section~\ref{sec:seller-informed}):

%We consider the buyer offering mechanism, where the buyer makes the best take-it-or-leave-it offer based on their signal and the seller decides whether to accept or reject the offer. In the extreme point $(1,1)$, we have that the buyer offering mechanism provides an $\frac e {e-1}$-approximation \cite{DS24}. We show that as the buyer's information increases, the approximation ratio of the buyer-offering mechanism improves accordingly:

\headline{Informal Theorem: }{Fix $1\geq \beta \geq 0$. Consider some mechanism where the seller always have a dominant strategy. Suppose that this mechanism provides an approximation ratio of $c$ in the $(1,1)$ private values model. For every $(1,\beta)$-information structure, there is a mechanism that provides an approximation ratio of $O(\frac c \beta)$. Moreover, there exists an $(1,\informedcoeffbuyer)$-information structure where no Bayesian incentive-compatible mechanism has an approximation ratio better than $\frac{2}{3\informedcoeffbuyer}$. } 

The approximation mechanism for the private values model can be, e.g., a posted price mechanism \cite{BD14, CKLT16, BD21, KPV21, LRW23, CW23} or the buyer offering mechanism \cite{DS24}. Each such mechanism provides a constant approximation ratio. Moreover, in the special case where the valuations satisfy the single-crossing condition, we strengthen our mechanism and show that a constant-factor approximation exists for this edge.

%Why is this interesting? posted prices work well in (1,1). The buyer offering mechanism was designed for the correlated case where posted prices fail to provide a good approximation. Yet, when the buyer is only slightly less informed then posted price no longer provide a good approximation ratio and we have to resort to the more flexible buyer offering mechanisms.

The last ``edge'' that we consider is the symmetric scenario where the buyer is fully informed and the seller is partially informed (Section~\ref{sec:buyer-informed}). Unlike the previous case (where the seller is fully informed and the buyer is partially informed), we now quickly encounter the impossibility barrier:

\headline{Informal Theorem: }{Fix $1 > \alpha >0$. There exists an $(\alpha,1)$-information structure where no Bayesian incentive-compatible mechanism provides a constant approximation ratio.}

In a sense, our results identify a fundamental challenge in designing mechanisms that approximate expected welfare in interdependent bilateral trade. The challenge stems, of course, from the necessity of each of the players to reassess their value conditioned on the other player accepting a trade. The specific point is that the strategic behavior of the buyer often reveals ``too much'' information about their private signal. That is, if the buyer accepts a trade then the seller can infer that the contribution of the buyer's signal to the seller's value is higher than the seller's ex-ante expected contribution of the signal. Thus, knowing that the buyer accepts a trade makes the seller less likely to do so. Similarly, if a seller accepts a trade, it means that the contribution of the seller's signal to the buyer's value is smaller than expected.

Our impossibilities leverage precisely this point---the buyer's strategic behavior leaks too much information---and show that this issue is unavoidable. Our positive results, however, are based on identifying cases where this leakage does not happen. One such example is the edge $(1,\informedcoeffbuyer)$ where the seller is fully informed and thus their estimation of their value is not affected by strategic calculations of the buyer. Our positive result for the edge $(\informedcoeffseller,0)$ is based on finding a posted price that the buyer always accepts. Thus, the estimation of the seller for their own value is also not affected by the strategic behavior of the buyer in this case as well.

Looking at Figure \ref{fig:intro-information-rectangle}, we see that the positive results that we got require the seller to either be fully informed or fully uninformed. In Section \ref{sec-interior} we prove that this is necessary: if the seller acquires or loses a tiny bit of information, the approximation ratio deteriorates significantly.

We note that there are some information structures that can be represented by several functions. For example, if $v_s(\signal{s}) = \signal{s}+1$, then both $\sellervaluationsellerfunc(x_s)=x_s, \sellervaluationbuyerfunc(x_s)=1$ and  $\sellervaluationsellerfunc(x_s)=x_s+0.5, \sellervaluationbuyerfunc(x_s)=0.5$ are valid choices that will give different values of $\alpha,\beta$. Of course, this makes our positive result only stronger as we can choose the decomposition to functions that will give the best results. As for our negative results, we make sure that all information structures in our impossibilities have a unique decomposition.
\subsection*{Approach II: Families of Information Structures}

While the primary focus of this paper is on quantifying the influence of signals on valuations, we also take a complementary approach by identifying meaningful families of information structures and analyzing their approximation ratios. This approach is common in the interdependent values literature in order to circumvent its many strong impossibility results.
Our first result in this direction concerns the linear setting:

\headline{Theorem:}{ Suppose that $\buyervaluationbuyerfunc, \buyervaluationsellerfunc, \sellervaluationbuyerfunc,\sellervaluationsellerfunc$ are linear. Then there is a posted price mechanism that is Bayesian incentive-compatible and provides a constant approximation to the social welfare.}

We also consider the more general case where the functions are polynomials of a bounded degree. Here, we prove both possibility and impossibility results: 

%\headline{Theorem: }
\headline{Informal Theorem:}
{ 
\begin{itemize}
\item Suppose that $\buyervaluationbuyerfunc, \buyervaluationsellerfunc, \sellervaluationbuyerfunc,\sellervaluationsellerfunc$ are polynomials with maximal degree $k$. There is a Bayesian incentive-compatible mechanism that provides an $O(k^2)$-approximation to the social welfare.
\item For every $k$, there are $\buyervaluationbuyerfunc, \buyervaluationsellerfunc, \sellervaluationbuyerfunc,\sellervaluationsellerfunc$ that are all polynomials of maximal degree $k$, such that no Bayesian incentive-compatible mechanism provides an approximation ratio better than $\Omega(k)$ to the social welfare.
\end{itemize}}

\subsection*{Future Work}

Our work raises many intriguing questions for future research. For example, we expect most points in the interior of the square depicted in Figure \ref{fig:intro-information-rectangle} to be red, indicating that they do not admit good approximation mechanisms. However, it would be interesting to find natural conditions or interesting families of information structures that do admit good approximations, similar to our findings for information structures with polynomial functions.

Additionally, while this paper focuses on welfare maximization in bilateral trade, another well-studied objective is maximizing the gains from trade \cite{M07,BD14,BCWZ17,DMSW22,F22,CaiGMZ20,BabaioffGG20, DKB}. A $c$-approximation for gains from trade directly implies a $c$-approximation for welfare maximization, hence all impossibility results established in this paper also apply to the objective of maximizing gains from trade. We leave open the question of whether there exist meaningful families of information structures that allow for non-trivial approximations of gains from trade.

%% file: Preliminaries.tex
\section{Preliminaries}\label{sec:preliminaries}

% In the bilateral trade problem with interdependent values, there are two agents: a seller, who owns the item, and a buyer, who seeks to purchase it. Each agent holds a signal, $\signal{s}$ for the seller and $\signal{b}$ for the buyer, and their valuations are public functions of both signals (see definition~\ref{def:valuation-functions})\footnote{The same way the distributions have to be public in bilateral trade, so do the valuation functions, as there are no prior free mechanisms that provide a good approximation.}.

In the bilateral trade problem with interdependent values, there are two agents: a seller, who owns the item, and a buyer, who seeks to purchase it. Each agent holds a signal, $\signal{s}$ for the seller and $\signal{b}$ for the buyer, and their valuations are public functions of both signals.\footnote{The same way the distributions have to be public in bilateral trade, so do the valuation functions, as there are no prior-free mechanisms that provide a good approximation.}

\begin{definition}[Valuation Functions]\label{def:valuation-functions}
We assume that the agents' valuation functions are separable in both signals, as defined by Equation~\eqref{equ:def-buyer-value} for the buyer and Equation~\eqref{equ:def-seller-value} for the seller.
\begin{subequations}
    \begin{equation}\label{equ:def-buyer-value}
        v_b(\signal{b}, \signal{s}) = \buyervaluationbuyerfunc{}(\signal{b}) + \buyervaluationsellerfunc(\signal{s}),
    \end{equation}
        \begin{equation}\label{equ:def-seller-value}
        v_s(\signal{b}, \signal{s}) = \sellervaluationbuyerfunc{}(\signal{b}) + \sellervaluationsellerfunc(\signal{s}).
    \end{equation}
\end{subequations}
\end{definition}

We assume that the functions $\buyervaluationbuyerfunc{}, \buyervaluationsellerfunc{}, \sellervaluationsellerfunc{}, \sellervaluationbuyerfunc{}$  are monotone non-decreasing \footnote{All our positive results hold without this assumption.} and that the signals are drawn independently from a uniform distribution over $[0,1]$.\footnote{Note that the assumption that the signals are distributed uniformly over $[0,1]$ is without loss of generality, as we don't restrict the functions, except for Section~\ref{sec:poly-function}.} Each tuple of $\buyervaluationbuyerfunc{}, \buyervaluationsellerfunc{}, \sellervaluationsellerfunc{}, \sellervaluationbuyerfunc{}$ defines an information structure.

 A \emph{direct mechanism} $\mechanism$ %for an information structure $(\buyervaluationbuyerfunc{}, \buyervaluationsellerfunc{}, \sellervaluationsellerfunc{}, \sellervaluationbuyerfunc{})$
 consists of two functions, $\mechanism = (\allocation, \payment)$ that receive agents' signals $\signal{s},\signal{b}$ as input, where $\allocation(\signal{b}, \signal{s})$  indicates the probability of trade, and $\payment(\signal{b}, \signal{s})$ is the price that the buyer pays the seller. We provide definitions for direct mechanisms, although we sometimes describe our mechanisms as indirect to make the presentation clearer. This is without loss of generality by the revelation principle. 

One restriction of the bilateral trade setting is of \emph{budget balance}, where the mechanism designer does not have to inject money. This may be formalized with distinct payment functions $\payment_b$ and $\payment_s$ for the buyer and seller respectively and requiring $\payment_b \geq \payment_s$. Because we use $\payment = \payment_b = \payment_s$, our mechanisms are all automatically \emph{strongly budget-balanced}: no money is left on the table.
 
The two other constraints of the mechanisms we design are Bayesian incentive-compatibility and interim individual rationality.

\begin{definition}[Bayesian Incentive Compatibility (BIC)]\label{def:BIC}
    A mechanism $\mechanism = (\allocation, \payment)$ is 
{\em Bayesian incentive-compatible} if for every $\signal{b},\signal{b}'$:
\begin{equation*}
\E_{\signal{s}}[x(\signal{b},\signal{s}) \cdot v_{b}(\signal{b},\signal{s}) - p(\signal{b},\signal{s})] \geq \E_{\signal{s}}[x(\signal{b}',\signal{s}) \cdot v_{b}(\signal{b},\signal{s}) - p(\signal{b}',\signal{s})],
%\label{eq:IC}
\end{equation*}
and for every $\signal{s},\signal{s}'$:
\begin{equation*}
\E_{\signal{b}}[p(\signal{b},\signal{s}) - x(\signal{b},\signal{s}) \cdot v_{s}(\signal{b},\signal{s})] \geq \E_{\signal{b}}[p(\signal{b},\signal{s}') - x(\signal{b},\signal{s}') \cdot v_{s}  x(\signal{b},\signal{s})].
%\label{eq:IC}
\end{equation*}
\end{definition}
\begin{definition}[Interim Individual Rationality (Interim IR)]\label{def:interim-IR}
    A mechanism $\mechanism = (\allocation, \payment)$ is 
{\em interim individually rational} if for every $\signal{b}$ and $\signal{s}$, 
\begin{equation*}
\E_{\signal{s}}[x(\signal{b},\signal{s}) \cdot v_{b}(\signal{b},\signal{s}) - p(\signal{b},\signal{s})] \geq 0 \hspace{1cm} \text{and} \hspace{1cm} \E_{\signal{b}}[p(\signal{b},\signal{s}) - x(\signal{b},\signal{s}) \cdot v_{s}(\signal{b},\signal{s})] \geq 0
\label{eq:IIR}
\end{equation*}
\end{definition}

For mechanisms that are BIC for the seller (respectively for the buyer), the classic Myersonian analysis can be applied. This allows us to retrieve a seller (buyer) payment formula, as well as argue that $\allocation$ and $\payment$ must be non-increasing functions of $\signal{s}$ (non-decreasing of $\signal{b}$). For a more in-depth discussion on Myersonian theory for interdependent values we refer to \cite{RT-C16}.

\begin{lemma}[Myersonian Seller Payment Formula]\label{lemma:seller-payment-formula}  
If a mechanism $(\allocation,\payment)$ is BIC for the seller then for every $\signal{s} \in [0,1]$ the seller's payment function must satisfy:
\begin{equation*}
    \mathbb{E}_{\signal{b}} \left[ \payment(\signal{b},\signal{s}) \right] = \mathbb{E}_{\signal{b}} \left[ \payment(\signal{b},0) + v_s(\signal{b},\signal{s}) \cdot \allocation(\signal{b},\signal{s}) - v_s(\signal{b},0) \cdot \allocation(\signal{b},0) - \int_{0}^{\signal{s}} {\allocation(\signal{b},z) \left. \frac{\partial v_s(\signal{b},\signal{s})}{\partial \signal{s}} \right|_{\signal{s}=z}} \, dz \right].
\end{equation*}
\end{lemma}

\begin{lemma}[Myersonian Buyer Payment Formula]\label{lemma:buyer-payment-formula}  
If a mechanism $(\allocation,\payment)$ is BIC for the buyer then for every $\signal{b} \in [0,1]$ the buyer's payment function must satisfy:
\begin{equation*}
\mathbb{E}_{\signal{s}} \left[ \payment(\signal{b},\signal{s}) \right] = \mathbb{E}_{\signal{s}} \left[ \payment(0,\signal{s}) + v_b(\signal{b},\signal{s}) \cdot \allocation(\signal{b},\signal{s}) - v_b(0,\signal{s}) \cdot \allocation(0,\signal{s}) - \int_{0}^{\signal{b}} {\allocation(z,\signal{s}) \left. \frac{\partial v_b(\signal{b},\signal{s})}{\partial \signal{b}} \right|_{\signal{b} = z}} \, dz \right].
\end{equation*}
\end{lemma}

We now introduce some notions that describe classes of information structures.

\begin{definition}[Single-Crossing]\label{def:single-crossing}
    A bilateral trade instance $(v_s,v_b)$ is \emph{Single-Crossing (SC)}  if for every signal profile $(\signal{s},\signal{b})$, 
    \begin{equation*}
    \frac{\partial v_b(\signal{b},\signal{s})}{\partial \signal{b}} \geq \frac{\partial v_{s}(\signal{b},\signal{s})}{\partial \signal{b}} \hspace{1cm} \text{and} \hspace{1cm} \frac{\partial v_s(\signal{b},\signal{s})}{\partial \signal{s}} \geq \frac{\partial v_{b}(\signal{b},\signal{s})}{\partial \signal{s}}.
    \label{eq:SC}
\end{equation*}
\end{definition}

\begin{definition}[Seller is $\informedcoeffseller$-informed]\label{def:seller-informed}
    A seller is $\informedcoeffseller$-informed, for $\informedcoeffseller \in [0,1]$, if it holds that $\informedcoeffseller = \frac{\mathbb{E}_{\signal{s}}[\sellervaluationsellerfunc(\signal{s})]}{\mathbb{E}_{\signal{s}}[\sellervaluationsellerfunc(\signal{s})]+ \mathbb{E}_{\signal{b}}[\sellervaluationbuyerfunc{}(\signal{b})]}$.
\end{definition}

% Similarly, for the buyer. 
\begin{definition}[Buyer is $\informedcoeffbuyer$-informed]\label{def:buyer-informed}
    A buyer is $\informedcoeffbuyer$-informed, for $\informedcoeffbuyer \in [0,1]$, if it holds that $\informedcoeffbuyer = \frac{\mathbb{E}_{\signal{b}}[\buyervaluationbuyerfunc{}(\signal{b})]}{\mathbb{E}_{\signal{b}}[\buyervaluationbuyerfunc{}(\signal{b})]+ \mathbb{E}_{\signal{s}}[\buyervaluationsellerfunc{}(\signal{s})]}$.
\end{definition}

When $\informedcoeffbuyer = 1$, we say that the buyer is fully informed, as their value is only affected by their own signal. Conversely, when $\informedcoeffbuyer =0$, we say that the buyer is uninformed, as their value is determined solely by the seller's signal. We use similar notation for the seller.

%\begin{remark} \label{remark:generalinformed}
%    All of our results apply to general valuation functions and not just additively separable ones by generalizing these definitions of $\alpha$ and $\beta$ as follows:
%    \begin{equation}
%    \label{eq:generalinformed}
%    \alpha = \frac{\E_{\signal{s}}[v_s(0_b,\signal{s})]}{ \E_{\signal{b},\signal{s}}[v_s(\signal{b},\signal{s})]}
%    \hspace{1cm} \text{and} \hspace{1cm}
%    \beta = \frac{\E_{\signal{b}}[v_b(\signal{b},0_s)]}{ \E_{\signal{b},\signal{s}}[v_b(\signal{b},\signal{s})]}.
%    \end{equation}
%\end{remark}

%\begin{fact}[OPT Bounds] \label{fact:optbound}
%    Optimal expected welfare is always lower-bounded by the welfare of the no-trade mechanism (the seller's expected value) and upper-bounded by the sum of both players' expected values:    $$\E_{\signal{b},\signal{s}}[v_s(\signal{b},\signal{s})] \leq OPT \leq \E_{\signal{b},\signal{s}}[v_b(\signal{b},\signal{s})] + \E_{\signal{b},\signal{s}}[v_s(\signal{b},\signal{s})].$$
%\end{fact}

% \asnote{I don't use this, I use Myerson over the valuations space, I'll explain it in the relevant sections}

% \ref{}
% \asnote{Would it better to say fully informed and completely uninformed?}
% \asnote{not sure if to mention uniform signals} \kgnote{Yes, when defining values/signals, mention it's without loss to map to $[0,1]$.}\asnote{added that}
% \kgnote{Still undefined: 
% \begin{itemize}
%     \item using a single $p$ function (what budget balance is and that we are only SBB)
%     \item what utility is for each agent
% \end{itemize}}

%% file: boundaries.tex
\section{A Fully Informed Seller: The $(1,\informedcoeffbuyer)$ Edge}\label{sec:seller-informed}

% We examine the case where the seller is fully informed, meaning $v_s(\signal{b}, \signal{s}) = \sellervaluationsellerfunc{}(\signal{s})$, while the buyer is $\informedcoeffbuyer$-informed, for some $\informedcoeffbuyer \in [0,1]$, specifically, $v_b(\signal{b}, \signal{s}) = \buyervaluationbuyerfunc{}(\signal{b}) + \buyervaluationsellerfunc{}(\signal{s})$. When $\informedcoeffbuyer = 0$, the seller is fully informed and the buyer is completely uninformed. In this case we have an impossibility result discussed in Section~\ref{sec:impossbility-seller-fully-informed-buyer-un}. Conversely, when $\informedcoeffbuyer =1$, both agents are fully informed, corresponding to the private value case, where a constant approximation result is known. For intermediate values $\informedcoeffbuyer \in (0,1]$, we show that the buyer-offering mechanism provides an approximation ratio that is linear in $\frac{1}{\informedcoeffbuyer}$ (Section~\ref{sec:edge-seller-informed}). This bound is tight in the sense that it approaches $\infty$ as $\informedcoeffbuyer$ approaches $0$.
% In the buyer-offering mechanism, the buyer, after observing his own signal $\signal{b}$, makes a take-it-or-leave-it offer to the seller.

In this section, we examine the case where the seller is fully informed and the buyer is $\informedcoeffbuyer$-informed for some $\informedcoeffbuyer \in [0,1]$. In subsection~\ref{sec:edge-seller-informed-algs}, we derive an $O(\frac{1}{\beta})$ approximation ratio for $\informedcoeffbuyer \in (0,1]$ by reducing to the private values case ($\informedcoeffbuyer=\informedcoeffseller=1$). Moreover, when the valuation functions satisfy the single-crossing condition (Definition~\ref{def:single-crossing}), we show a constant approximation is achievable regardless of the buyer's level of informedness (Corollary~\ref{cor:single-crossing-constant-approximation}). We complement these results by showing that if the single-crossing condition does not hold, then no mechanism has an approximation ratio better than $\frac{2}{3\informedcoeffbuyer}$. In particular, for $\informedcoeffbuyer = 0$ (i.e., the seller is fully informed but the buyer is fully uninformed) no constant-factor mechanism exists (Subsection~\ref{sec:impossbility-seller-fully-informed-buyer-un}).
% \asnote{Maybe we want to say that we show that the approximation ratio of $\frac{1}{\informedcoeffbuyer}$ is tight? we show a lower bound of $\Omega(\frac{1}{\informedcoeffbuyer})$}

\subsection{Mechanisms for the $(1,\informedcoeffbuyer)$ Edge} \label{sec:edge-seller-informed-algs}

In this section, we show a reduction from the case of a fully informed seller and $\informedcoeffbuyer$-informed buyer to the private values case.

% one, and we prove that given a one-sided DSIC mechanism for the seller, for private values, we can turn it into a BIC mechanism for the case of fully informed seller  and $\informedcoeffbuyer$-informed buyer

\newcommand{\appratio}{\gamma}
\newcommand{\instansepv}{\mathcal{I_{\text{PV}}}}

\begin{theorem}\label{thm:possibility-seller-informed}
    Let $\mechanism$ be a mechanism for the private values case that is one-sided DSIC for the seller and has an approximation ratio of $\appratio$. Consider an information structure with $\informedcoeffbuyer > 0$ and a fully informed seller ($\informedcoeffseller =1$). Then there exists a mechanism that is BIC for the buyer and DSIC for the seller with an approximation ratio of $\frac{2\appratio}{\informedcoeffbuyer}$. 
\end{theorem}

\begin{proof}[Proof of Theorem~\ref{thm:possibility-seller-informed}]
Let $I_o$ be an information structure $\mathcal{I}_{o} = (\buyervaluationbuyerfunc{}, \buyervaluationsellerfunc{}, \sellervaluationsellerfunc{}, \sellervaluationbuyerfunc{})$, where
the buyer is $\informedcoeffbuyer$-informed and the seller is fully informed, i.e., $\sellervaluationbuyerfunc{}(\signal b)=0$ for all $\signal b$, and
$\E_{\signal{s}}[\buyervaluationsellerfunc{}(\signal{s})] = \left(\frac{1-\informedcoeffbuyer}{\informedcoeffbuyer}\right)\cdot\E_{\signal{b}}[\buyervaluationbuyerfunc{}(\signal{b})]$.

We divide the analysis into two. If $\E_{\signal{s}}[v_s(\signal{s})] \geq {\E_{{\signal{s}, \signal{b}}}}[v_b(\signal{b}, \signal{s})]$, not trading the item results in an approximation ratio of at most 2. In the rest of the proof, we consider the second case, where $\E_{\signal{s}}[v_s(\signal{s})] \leq {\E_{{\signal{s}, \signal{b}}}}[v_b(\signal{b}, \signal{s})]$.   

We define the following private values information structure $\instansepv$, where the buyer's value is $v_b(\signal{b}) = \buyervaluationbuyerfunc{}(\signal{b})$ and the seller's value is $v_s(\signal{s}) = \sellervaluationsellerfunc{}(\signal{s})$.
% We now apply the private values mechanism $\mechanism$ to $\instansepv$. \sdnote{not sure what it means to apply a mechanism to an information structure. do we really need this sentence?} \asnote{sure, it doesn't add anything, we can drop it}
Note that since $\mechanism$ is one-sided DSIC for the seller, for every value of the buyer $\buyervaluationbuyerfunc{}(\signal{b})$, the seller faces a take-it-or-leave-it offer in $\instansepv$.
For every $\signal{b}$, let $p_b(\signal{b})$ be the offer given to the seller in $\instansepv$ according to $\mechanism$. 

We claim that there exists a BIC mechanism $\mechanism_o$ for $\mathcal{I}_o$ such that, for every $\signal{b}$, the seller is faced with a take-it-or-leave-it offer $p'_b(\signal{b})$ that is at least as high as the offer $p_b(\signal{b})$ in the private values information structure $\instansepv$ for the same realization of $\signal{b}$.
Since the seller is fully informed, they accept if and only if their value is at most the offer, regardless of whether we consider the private values information structure or the original one. The value of the buyer in $\mathcal{I}_o$ is at least their value in $\instansepv$, for every $\signal b$. As a result, they may make a different offer $p'_b(\signal{b})$ for the seller. However, this offer can only be larger than $p_b(\signal{b})$, which in turn can only increase the welfare. We thus have:

\begin{corollary}\label{corollary-reducion-pv-offer-is-larger}
    For every $\signal{b}$, it holds that $p'_b(\signal{b}) \geq p_b(\signal{b})$.
    % let $p_b$ be the take-it-or-leave-it offer for the seller in $\instansepv$ according to $\mechanism$. Let $p'_b$ be another offer made by $\mechanism$ given a different signal $\signal{b}'$. Then, if the buyer prefers offer $p'_b$ over offer $p_b$ in the original information structure when his signal is $\signal{b}$, then $p_b \leq p'_b$. 
\end{corollary}

We turn to proving the approximation guarantee. Since $\mechanism$ provides an approximation ratio of $\appratio$ for $\mathcal I_{PV}$, the welfare of $\mechanism$ is at least the expected value of the buyer in $\instansepv$ divided by the approximation ratio, i.e., $\frac{1}{\appratio}\cdot \E_{\signal{b}}[\buyervaluationbuyerfunc{}(\signal{b})]$. 
By Corollary~\ref{corollary-reducion-pv-offer-is-larger}, the welfare of the resulting mechanism, $\mechanism_o$, for $\mathcal{I}_o$ is at least the welfare of $\mechanism$ for $\instansepv$. Combining everything, we get that the approximation ratio of the resulting mechanism $\mechanism_o$ for $\mathcal{I}_o$ is: 

$$
\frac{OPT}{ALG} \leq \frac{\E_{\signal{s}}[v_s(\signal{s})] + \E_{\signal{b}}[\buyervaluationbuyerfunc{}(\signal{b})] + \E_{\signal{s}}[\buyervaluationsellerfunc{}(\signal{s})]}{\frac{1}{\appratio}\cdot \E_{\signal{b}}[\buyervaluationbuyerfunc{}(\signal{b})]} \leq \frac{2( \E_{\signal{b}}[\buyervaluationbuyerfunc{}(\signal{b})] + \E_{\signal{s}}[\buyervaluationsellerfunc{}(\signal{s})])}{\frac{1}\appratio{}\cdot \E_{\signal{b}}[\buyervaluationbuyerfunc{}(\signal{b})]} \leq \frac{\frac{2}{\informedcoeffbuyer}\cdot\E_{\signal{b}}[\buyervaluationbuyerfunc{}(\signal{b})]}{\frac{1}{\appratio}\cdot \E_{\signal{b}}[\buyervaluationbuyerfunc{}(\signal{b})]} = \frac{2\appratio}{\informedcoeffbuyer}
$$
where the second inequality follows because in this case we assume that $\E_{\signal{s}}[v_s(\signal{s})] \leq \E_{\signal{s}, \signal{b}}[v_b(\signal{b} \signal{s})]$.

We conclude the proof of the theorem by providing the proof of Corollary~\ref{corollary-reducion-pv-offer-is-larger}.

\begin{proof}[Proof of Corollary~\ref{corollary-reducion-pv-offer-is-larger}]    
Fix a realization of $\signal{b}$. Since $\mechanism$ is BIC for the buyer in the private values case, the buyer's profit from an offer  $p_b(\signal{b})$ to the seller is at least as large as their profit from an offer $p_b(\signal{b'})$ to the seller. That is,
\begin{equation}\label{ineq:higher-val-higher-price}
\Pr[v_s(\signal{s})\leq p_b(\signal{b})]\cdot\left(\buyervaluationbuyerfunc{}(\signal{b}) - p_b(\signal{b})\right) \geq \Pr[v_s(\signal{s})\leq p_b(\signal{b'})]\cdot\left(\buyervaluationbuyerfunc{}(\signal{b}) - p_b(\signal{b'})\right) .
\end{equation}
Now, if in $\mathcal{I}_o$, the buyer has a higher expected profit from an offer $p_b(\signal{b'})$ to the seller than from an offer $p_b(\signal{b})$, then we have:
\begin{multline*}
\Pr[v_s(\signal{s})\leq p_b(\signal{b})]\cdot\left(\buyervaluationbuyerfunc{}(\signal{b}) + \E_{\signal{s}}[\buyervaluationsellerfunc{}(\signal{s})| \, v_s(\signal{s}) \leq p_b(\signal{b})] - p_b(\signal{b})\right) \leq \\
\Pr[v_s(\signal{s})\leq p_b(\signal{b'})]\cdot\left(\buyervaluationbuyerfunc{}(\signal{b})  + \E_{\signal{s}}[\buyervaluationsellerfunc{}(\signal{s})| \, v_s(\signal{s}) \leq p_b(\signal{b'})] - p_b(\signal{b'})\right).
\end{multline*}
By Inequality~\ref{ineq:higher-val-higher-price}, it follows that
$$
\Pr[v_s(\signal{s})\leq p_b(\signal{b})]\cdot\E_{\signal{s}}[\buyervaluationsellerfunc{}(\signal{s})| \, v_s(\signal{s}) \leq p_b(\signal{b})] \leq \Pr[v_s(\signal{s})\leq p_b(\signal{b'})]\cdot\E_{\signal{s}}[\buyervaluationsellerfunc{}(\signal{s})| \, v_s(\signal{s}) \leq p_b(\signal{b'})].
$$
This is equivalent to
$$
\E_{\signal{s}}[\buyervaluationsellerfunc{}(\signal{s}) \cdot \indicator_{[v_s(\signal{s})\leq p_b(\signal{b})]}] \leq \E_{\signal{s}}[\buyervaluationsellerfunc{}(\signal{s}) \cdot \indicator_{[v_s(\signal{s})\leq p_b(\signal{b'})]}]
$$
which is only possible if $p_b(\signal{b}) \leq p_b(\signal{b'})$. Hence in $\mathcal{I}_o$, the buyer can only prefer higher prices than those in $\instansepv$, implying that $p_b'(\signal{b})\geq p_b(\signal{b})$.
\end{proof}
\end{proof}

\begin{corollary}\label{cor:single-crossing-constant-approximation}
    If the seller is fully informed, and the valuation functions satisfy the single-crossing condition, there exists a mechanism with a constant approximation ratio.
\end{corollary}

\newcommand{\crossingvalue}{t}

\begin{proof}[Proof of Corollary~\ref{cor:single-crossing-constant-approximation}]
    Under the single-crossing condition, if there exists a point $t$ where $v_s(t) \geq \buyervaluationsellerfunc{}(t)$, then for every $\signal{s}\geq t$, it holds that $v_s(\signal{s}) \geq \buyervaluationsellerfunc{}(\signal{s})$.
    If $\buyervaluationsellerfunc{}(0) > 0$, we normalize the information structure by defining:  
    $$\buyervaluationsellerfunc{}'(\signal{s}) = \buyervaluationsellerfunc{}(\signal{s}) - \buyervaluationsellerfunc{}(0), \qquad \buyervaluationbuyerfunc{}{}'(\signal{b}) = \buyervaluationbuyerfunc{}(\signal{b}) + \buyervaluationsellerfunc{}(0), \qquad \informedcoeffbuyer' = \frac{\E_{\signal{b}}[\buyervaluationbuyerfunc{}'(\signal{b})]}{\E_{\signal{b}}[\buyervaluationbuyerfunc{}'(\signal{b})] + \E_{\signal{s}}[\buyervaluationsellerfunc{}'(\signal{s})]}.$$
    Since the single-crossing condition holds for $v_s$ and $\buyervaluationsellerfunc{}$, it also holds for $v_s$ and $\buyervaluationsellerfunc{}'$.
    Moreover, since $v_s(0)\geq 0$ and $\buyervaluationsellerfunc{}'(0)=0$, it follows that $v_s(\signal{s}) \geq \buyervaluationsellerfunc{}'(\signal{s})$ for every $\signal{s} \geq 0$.
    We now claim that either not trading the item results in a good approximation or applying the mechanism from Theorem~\ref{thm:possibility-seller-informed} does.
    The approximation ratio when not trading the item is at most:
    $$
    \frac{OPT}{\E_{\signal{s}}[v_s(\signal{s})]} \leq \frac{\E_{\signal{s}}[v_s(\signal{s})]+ \E_{\signal{b}}[v_b(\signal{b})]}{\E_{\signal{s}}[v_s(\signal{s})]} = 1+ \frac{{\E_{\signal{b}}[\buyervaluationbuyerfunc{}'(\signal{b})] + \E_{\signal{s}}[\buyervaluationsellerfunc{}'(\signal{s})]}}{\E_{\signal{s}}[v_s(\signal{s})]} = 1+ \frac{\frac{1}{1-\informedcoeffbuyer'}\cdot\E_{\signal{s}}[\buyervaluationsellerfunc{}'(\signal{s})]}{\E_{\signal{s}}[v_s(\signal{s})]} \leq 1+ \frac{1}{1-\informedcoeffbuyer'}
    $$
    while the approximation ratio obtained from applying Theorem~\ref{thm:possibility-seller-informed} to the normalized information structure is $\frac{2\appratio}{\informedcoeffbuyer'}$, where $\appratio$ is the approximation ratio of a mechanism for the private values case. For instance, consider applying a posted price mechanism with an approximation ratio of $2$ (\cite{BD14}). The approximation guarantee follows since $\min\{1+ \frac{1}{1-\informedcoeffbuyer'}, \frac{4}{\informedcoeffbuyer'}\} \leq 5.5$.
\end{proof}

\subsection{An Impossibility for the $(1,\informedcoeffbuyer)$ Edge}\label{sec:impossbility-seller-fully-informed-buyer-un}

We now provide an impossibility result for this edge, complementing the approximation ratio obtained in Section~\ref{sec:edge-seller-informed-algs}. The formal statement of this impossibility is stated in the second part of Theorem~\ref{thm:uninformed-buyer-impossibility} and follows from the the first part. The first part is also used to obtain Theorem~\ref{thm:impossibility-alpha-0-edge}.

\begin{theorem}\label{thm:uninformed-buyer-impossibility}The following two statements hold:
\begin{enumerate}
\item For every $k, m \in \mathbb{N}$, there exists an information structure where the seller is $\informedcoeffseller$-informed and the buyer is $\informedcoeffbuyer$-informed such that $\informedcoeffseller > 0, \informedcoeffbuyer < 1, k\geq 4, m\geq 5$, and no BIC and interim IR mechanism can have an approximation ratio better than:
    \begin{equation*}
\frac{\informceoffbuyer+1}{\frac{1}{k}(1+ \informceoff)+ \informceoffbuyer  +\frac{4}{m}+\informceoffbuyer\cdot \frac{2}{k}+\frac{2}{k^2}\cdot\informceoff},
    \end{equation*}
    where $\informceoffbuyer =\frac{\informedcoeffbuyer}{1-\informedcoeffbuyer}$ and $ \sellerratioinformed = \frac{1-\informedcoeffseller}{\informedcoeffseller}$.
\item As a corollary. For every $\informedcoeffbuyer \in (0,1)$, there exists an information structure where the seller is fully informed and the buyer is $\informedcoeffbuyer$-informed, and no BIC and interim IR mechanism can provide an approximation ratio better than $\frac{2}{3\informedcoeffbuyer}$. Moreover, for $\informedcoeffbuyer =0$, for every $c>1$, there exists an information structure where the seller is fully informed and the buyer is uninformed and no BIC and interim IR mechanism can attain an approximation ratio better than $c$.
% for every $\informedcoeffseller > 0$, and every $c>1$, there exists an information structure where the buyer is totally uninformed and the seller is $\informedcoeffseller$-informed, and no BIC and interim IR mechanism provides an approximation ratio better than $c$.
\end{enumerate}
\end{theorem}

% \begin{theorem}\label{thm:imposs:seller:informed}
%     Fix $\informedcoeffbuyer \in (0,1)$. There exists an information structure where the seller is fully informed and the buyer is $\informedcoeffbuyer$-informed for which no BIC and interim IR mechanism can provide an approximation ratio better than $\frac{2}{3\informedcoeffbuyer}$. Moreover, for $\informedcoeffbuyer =0$, for every $c>1$, there exists an information structure where the seller is fully informed and the buyer is uninformed and no BIC and interim IR mechanism can attain an approximation ratio better than $c$.
% \end{theorem}

The first part of the theorem implies the second one. For $\informedcoeffbuyer \in (0,1)$, let $\informedcoeffseller=1, \, m=4k, \,  k=\lceil\frac{4}{\informedcoeffbuyer}\rceil$. Then for the information structure $\mathcal{I}_{k,m, \informedcoeffseller, \informedcoeffbuyer}$ there is no BIC and interim IR mechanism with an approximation ratio better than:
\begin{align*}
\frac{\informceoffbuyer+1}{\frac{1}{k}(1+ \informceoff)+ \informceoffbuyer  +\frac{4}{m}+\informceoffbuyer\cdot \frac{2}{k}+\frac{2\informceoff}{k^2}} 
% = \frac{\frac{1}{1-\informedcoeffbuyer}}{\frac{1}{k}+ \informceoffbuyer  +(1-\frac{1}{k})\cdot(\frac{4}{m}+\informceoffbuyer\cdot \frac{2}{k})}
= \frac{\frac{1}{1-\informedcoeffbuyer}}{\frac{1}{k}+ \informceoffbuyer  +\frac{1}{k}+\informceoffbuyer\cdot \frac{2}{k}} = \frac{\frac{1}{1-\informedcoeffbuyer}}{\frac{2}{k}(\frac{1}{1-\informedcoeffbuyer})  +\frac{\informedcoeffbuyer}{1-\informedcoeffbuyer}} \geq \frac{\frac{1}{1-\informedcoeffbuyer}}{\frac{\informedcoeffbuyer}{1-\informedcoeffbuyer}\cdot(\frac{3}{2})}=\frac{2}{3\informedcoeffbuyer}.
\end{align*}

For $\informedcoeffbuyer =0$, we get $\informceoffbuyer = 0$,  for every $c\geq 2$,  let $\informedcoeffseller=1, \, k = \lceil 2\cdot c \rceil, \, m=4k^2$. Now, the approximation ratio of every BIC and interim IR mechanism for $\mathcal{I}_{k,m,\informedcoeffseller, \informedcoeffbuyer}$ is at least:

\begin{align*}
\frac{\informceoffbuyer+1}{\frac{1}{k}(1+ \informceoff)+ \informceoffbuyer  +\frac{4}{m}+\informceoffbuyer\cdot \frac{2}{k}+\frac{2\informceoff}{k^2}} = \frac{1}{\frac{1}{k}+ \frac{4}{m}} = \frac{1}{\frac{2}{k}} \geq c.
\end{align*}

% The proof of the first part of Theorem~\ref{thm:uninformed-buyer-impossibility} parallels that of Theorem~\ref{thm:buyer-informed-inpossibility}. We defer its proof to Appendix~\ref{sec:app-uninformed-buyer}.

We defer the proof of the first part of Theorem~\ref{thm:uninformed-buyer-impossibility} to Appendix~\ref{sec:app-uninformed-buyer}.

\section{A Fully Informed Buyer: The $(\informedcoeffseller,1)$ Edge} \label{sec:buyer-informed}

This section considers a fully informed buyer, i.e., $v_b(\signal{b}, \signal{s}) = \buyervaluationbuyerfunc{}(\signal{b})$, and a partially-informed seller. When the seller is fully informed ($\informedcoeffseller = 1$), we are at the extreme point $(1,1)$, which aligns with the private values setting where posted price mechanisms provide a constant approximation ratio. The main result in this section shows that the $(1,1)$ extreme point is unique among all points on the edge $(\informedcoeffseller,1)$: in no other point can we always achieve a constant approximation to welfare.

The formal statement of this impossibility is stated in the second part of Theorem \ref{thm:buyer-informed-inpossibility} and is an almost immediate corollary of the first part. The first part is also used in order to obtain Theorem \ref{thm:imposs-uninformed-seller}.

% \kgc{So these results are not only for the edge. They're also for the interior. Yes?}\asnote{Yes, also for the second construction} \kgc{So update section title / intro / square pic accordingly?}\asnote{still working on it, it doesn't give interesting results for every point in the interior}

\begin{theorem}\label{thm:buyer-informed-inpossibility}The following two statements hold:
\begin{enumerate}
    \item For every $k, m \in \mathbb{N}$, there exists an information structure where  the buyer is $\informedcoeffbuyer$-informed and the seller is $\informedcoeffseller$-informed such that $ \informedcoeffseller < 1, \informedcoeffbuyer > 0, k\geq 2, m\geq 3$, and no BIC and interim IR mechanism has an approximation ratio better than:
    \begin{equation*}
        \frac{\frac{1}{\informedcoeffbuyer}}{\frac{1}{k}\frac{1}{1-\informedcoeffseller} + \frac{1-\informedcoeffbuyer}{\informedcoeffbuyer} + 2k\frac{1}{m}+2k\frac{1-\informedcoeffbuyer}{\informedcoeffbuyer}}
    \end{equation*}
    \item As a corollary, for every $0\leq \informedcoeffseller < 1$ and $c>1$, there exists an information structure where the buyer is fully informed and the seller is $\informedcoeffseller$-informed, where no BIC and interim IR mechanism provides an approximation ratio better than $c$. Moreover, this information structure satisfies the single-crossing condition.
\end{enumerate}
\end{theorem}

% \asnote{change the structure, so the information structure def is integrated in the proof} 

Given $\informedcoeffseller$ and $c$, to prove the second part of the theorem, we choose parameters $k = \lceil c \cdot(\frac{2-\informedcoeffseller}{1-\informedcoeffseller})\rceil, \, m=2k^2$ such that, by the first part, there exists an information structure $\mathcal{I}_{k,m,\informedcoeffseller, \informedcoeffbuyer}$ where the buyer is fully informed ($\informedcoeffbuyer = 1$) and the seller is $\informedcoeffseller$-informed, and no BIC and interim IR mechanism can attain an approximation ratio better than: 
% $c$-approximation. The formal proof appears below (Section~\ref{sec:construction:buyer:has:signal}), but we first provide some intuition for the proof.

    $$
    \frac{\frac{1}{\informedcoeffbuyer}}{\frac{1}{k}\frac{1}{1-\informedcoeffseller} + \frac{1-\informedcoeffbuyer}{\informedcoeffbuyer} + 2k\frac{1}{m}+2k\frac{1-\informedcoeffbuyer}{\informedcoeffbuyer}} = \frac{1}{\frac{1}{k}\frac{1}{1-\informedcoeffseller} + 2k\frac{1}{m}} \geq c.
    $$

Before we prove the first part of Theorem~\ref{thm:buyer-informed-inpossibility}, we provide some intuition to the case described in the second part. 
In the information structure $\mathcal{I}_{k, m,\informedcoeffseller, \informedcoeffbuyer}$, the buyer has an informative signal $\signal{b}$. This signal is a discrete random variable with support $[m]\cup \{0\}$. The buyer's valuation function is defined as $v_b(\signal{b}, \signal{s}) = \buyervaluationbuyerfunc{}(\signal{b})$, in this case the buyer is fully informed, hence $\buyervaluationsellerfunc{}(\signal{s})=0$. The seller's valuation is a function of the buyer's signal plus a constant: $v_s(\signal{b}, \signal{s}) = \sellervaluationbuyerfunc{}(\signal{b}) + \sellervaluationsellerfunc{}(\signal{s})$, where $\sellervaluationsellerfunc{}(\signal{s})$ is a constant function $\gamma \cdot \mu$ and $\mu = \frac{m}{k}$. We now define the signal probability mass function $\pdfbuyersignal(\signal{b})$, the buyer's valuation function $v_b(\signal{b})$, and the seller's valuation function $v_s(\signal{b})$:
    % Throughout this proof, we use $i$ to denote the values of $\signal{b}$, such that $v_b(i) = \buyervaluationbuyerfunc{}(i)$, and $v_s(i) = \sellervaluationbuyerfunc{}(i) + \reversecoeff\cdot\mu$
    \begin{equation*}\label{k-k-welfare-prob-values}
    \pdfbuyersignal(\signal{b}) = \begin{cases}
                 \frac{1}{k^{\signal{b}}} &  \signal{b} \in [m];\\
        1-\sum_{i=1}^m\frac{1}{k^i}  & \signal{b}=0.
           \end{cases} \quad \quad
    v_b(\signal{b}) = \begin{cases}
                 k^{\signal{b}}  &   \signal{b} \in [m] ;\\
                 0  &  \signal{b} \notin [m].
           \end{cases} \quad \quad
    v_s(\signal{b}) = \begin{cases}
                 k^{\signal{b}-1} +\reversecoeff \cdot \mu  &   \signal{b} \in [m] ;\\
                 \reversecoeff \cdot \mu &  \signal{b} \notin [m].
           \end{cases}           
\end{equation*}
    To simplify the presentation, the signal is not distributed uniformly over $[0,1]$ in the information structure $\mathcal{I}_{k, m,\informedcoeffseller, \informedcoeffbuyer}$. However, we can easily alter the information structure to ensure that the signals are distributed uniformly.\footnote{To do that, the buyer's valuation function will be a step function. Let $q_0, q_1, \dots{}, q_m, q_{m+1}$ be defined such that $q_j = \sum_{i=0}^j \pdfbuyersignal(j-1)$ for every $j \in [m]$, with $q_0=0$ and $q_{m+1} =1$.  Suppose the buyer's signal $x$ is drawn from a uniform distribution on $[0,1]$. Then, the buyer's valuation function is given by $k^j$ when $x \geq q_1$, where $j$ is the index satisfying $q_j\leq x\leq q_{j+1}$, and by $0$ when $0\leq x<q_1$. Similar construction can be applied to the seller's valuation.} We also note that the decomposition to functions $\sellervaluationsellerfunc, \sellervaluationbuyerfunc,\buyervaluationsellerfunc, \buyervaluationbuyerfunc$ is not necessarily unique, but the functions can be easily altered to obtain a unique decomposition.\footnote{For example, we can change the definition of $\sellervaluationsellerfunc$ (and, similarly, the other functions) at a single point $\sellervaluationsellerfunc(0)=0$. The decomposition is now unique and the analysis remains the same since $x_s=0$ with probability $0$.}

The theorem states that for some parameters -- in our case,  $\reversecoeff =   \frac{\informedcoeffseller}{1-\informedcoeffseller}>0, \, k = 10\cdot c \cdot(\reversecoeff+1), \, m=2k(k-1)$ -- no BIC and interim IR mechanism provides an approximation ratio better than $c$ for the $(\informedcoeffseller,1)$-information structure $\mathcal{I}_{k, m, \informedcoeffseller, \informedcoeffbuyer}$.

% \begin{theorem}\label{thm:buyer-informed-inpossibility}
%     Fix $0\leq \informedcoeffseller < 1$ and $c>1$. Let $\reversecoeff =   \frac{\informedcoeffseller}{1-\informedcoeffseller}>0$. Let $k, m \in \mathbb{N}$ be such that  $k = 10\cdot c \cdot(\reversecoeff+1), \, m=2k(k-1)$. 
%     No Bayesian incentive-compatible and interim individually rational mechanism provides an approximation ratio better than $c$ for the $(\informedcoeffseller,1)$-information structure $\mathcal{I}_{k, m}$. Moreover, this information structure satisfies the single crossing condition.
% \end{theorem}

In this intuitive discussion, we consider the extreme case where the seller is fully uninformed and the buyer is fully informed, i.e., the point $(0,1)$. Intuitively, at this point, only the buyer has an informative signal and is therefore the only one who can, to some extent, dictate the chosen outcome.
We start by considering a mechanism where the buyer is allowed to make any take-it-or-leave-it offer. Since the value of the seller is fully determined by $\signal{b}$, the buyer fully knows the value of the seller. Naively, the offer that the buyer makes can be equal to the value of the seller, and the seller---who is completely uninformed about their own value---should always accept. However, it is easy to see that this is not an equilibrium: if the seller always accepts, then the buyer is better off offering a lower price. 

More generally, since the seller has no information, the buyer can dictate the chosen equilibrium. The buyer will dictate the best equilibrium according to $\signal b$. Thus, using again the assumption that the seller is fully uninformed, each equilibrium $e$ boils down to a pair $(q_e,p_e)$, where $q_e$ is the probability of trade and $p_e$ is the trade price. The buyer will always choose the equilibrium $e$ that maximizes $(v_b-p_e)\cdot q_e$. 

To simplify this intuitive discussion, assume now that the mechanism is deterministic, i.e., that the probability of trade $q_e$ is either $1$ or $0$ for every equilibrium $e$. This limits the set of possible outcomes of the mechanism: the buyer will either choose the equilibrium $e$ where $p_e$ is minimal or will choose not to trade the item at all (if their value is less than $p_e$). Overall, we have that every deterministic mechanism either trades the item at some fixed price or does not trade the item at all. The information structure $\mathcal I_{k,m, \informedcoeffseller, \informedcoeffbuyer}$ is designed to ensure that both not trading the item and using any fixed price, each leads to a poor approximation ratio. Specifically, in $\mathcal{I}_{k,m, \informedcoeffseller, \informedcoeffbuyer}$, for all ``relevant'' values of $p$, if $v_b \geq p$, then $E[v_s|v_b \geq p]<p$. This property ensures that fixed price mechanisms cannot simultaneously be interim IR for the seller and provide a good approximation in $\mathcal{I}_{k,m, \informedcoeffseller, \informedcoeffbuyer}$.

We now proceed to turn this intuition into a full and formal proof of Theorem \ref{thm:buyer-informed-inpossibility} that in particular applies to all levels of informedness of the seller.

\subsection{Proof of Part I of Theorem \ref{thm:buyer-informed-inpossibility}}\label{sec:construction:buyer:has:signal}

% In the information structure $\mathcal{I}_{k, m}$, the buyer has an informative signal $\signal{b}$. This signal is a discrete random variable with support $[m]\cup \{0\}$. The buyer's valuation function is defined as $v_b(\signal{b}, \signal{s}) = \buyervaluationbuyerfunc{}(\signal{b})$. The seller's valuation is a function of the buyer's signal plus a constant: $v_s(\signal{s}, \signal{b}) = \sellervaluationbuyerfunc{}(\signal{b}) + \sellervaluationsellerfunc{}(\signal{s})$, where $\sellervaluationsellerfunc{}(\signal{s})$, is a constant function $\gamma \cdot \mu$, where $\mu = \frac{m}{k}$.

\begin{proof}
We construct the information structure $\mathcal{I}_{k,m,\informedcoeffseller, \informedcoeffbuyer}$. In the information structure $\mathcal{I}_{k, m, \informedcoeffseller, \informedcoeffbuyer}$, the buyer has an informative signal $\signal{b}$. This signal is a discrete random variable with support $[m]\cup \{0\}$. The seller has a less informative signal $\signal{s}$.  The buyer's valuation function is defined as $v_b(\signal{b}, \signal{s}) = \buyervaluationbuyerfunc{}(\signal{b}) + \buyervaluationsellerfunc{}(\signal{s})$. The seller's valuation is a function of the buyer's signal plus a constant: $v_s(\signal{b}, \signal{s}) = \sellervaluationbuyerfunc{}(\signal{b}) + \sellervaluationsellerfunc{}(\signal{s})$, where $\sellervaluationsellerfunc{}(\signal{s})$, is a constant function $\gamma \cdot \mu$, with $\mu = \E_{\signal{b}}[\sellervaluationbuyerfunc{}(\signal{b})], \, \informceoff=\frac{\informedcoeffseller}{1-\informedcoeffseller}$, and here $\mu = \frac{m}{k}$. 
\begin{equation*}
    \pdfbuyersignal(\signal{b}) = \begin{cases}
                 \frac{1}{k^{\signal{b}}} &  \signal{b} \in [m];\\
        1-\sum_{i=1}^m\frac{1}{k^i}  & \signal{b}=0.
           \end{cases} \hspace{.35cm}
    v_b(\signal{b},\signal{s}) = \begin{cases}
                 k^{\signal{b}} + \buyervaluationsellerfunc{}(\signal{s})  &   \signal{b} \in [m] ;\\
                 \buyervaluationsellerfunc{}(\signal{s})  &  \signal{b} \notin [m].
           \end{cases} \hspace{.35cm}
    v_s(\signal{b}) = \begin{cases}
                 k^{\signal{b}-1} +\reversecoeff \cdot \mu &   \signal{b} \in [m] ;\\
                 \reversecoeff \cdot \mu &  \signal{b} \notin [m].
           \end{cases}           
\end{equation*}

%To simplify the presentation, the signal $\signal{b}$ is not distributed uniformly over $[0,1]$ in the information structure $\mathcal{I}_{k, m, \informedcoeffseller, \informedcoeffbuyer}$. However, we can easily alter the information structure to ensure that the signals are distributed uniformly.\footnote{To do that, the buyer's valuation function will be a step function. Let $q_0, q_1, \dots{}, q_m, q_{m+1}$ be defined such that $q_j = \sum_{i=0}^j \pdfbuyersignal(j-1)$ for every $j \in [m]$, with $q_0=0$ and $q_{m+1} =1$.  Suppose the buyer's signal $x$ is drawn from a uniform distribution on $[0,1]$. Then, the buyer's valuation function is given by $k^j$ when $x \geq q_1$, where $j$ is the index satisfying $q_j\leq x\leq q_{j+1}$, and by $0$ when $0\leq x<q_1$. Similar construction can be applied to the seller's valuation.} 
Let $\mechanism =(\allocation, \payment)$ be a BIC and interim IR mechanism for the information structure $\mathcal{I}_{k,m,\informedcoeffseller, \informedcoeffbuyer}$.  To prove the theorem, we apply Lemma~\ref{lemma:buyer-has-signal-impossibility-small-alloc} to show that any BIC and interim IR mechanism for 
$\mathcal{I}_{k,m,\informedcoeffseller, \informedcoeffbuyer}$ must have a relatively small expected sum of allocation function values. We then use this result to bound the best possible approximation ratio for $\mathcal{I}_{k,m,\informedcoeffseller, \informedcoeffbuyer}$.
First, we prove the theorem using Lemma~\ref{lemma:buyer-has-signal-impossibility-small-alloc}, and then we turn to proving the lemma itself.

% The proof follows from Lemma~\ref{lemma:buyer-has-signal-impossibility-small-alloc} 

% and \ref{lemma:buyer-signal-tec-summation}. The first lemma shows that if $\mechanism$
% is interim IR, then the sum of allocation probabilities $\sum_{\signal{b}=1}^m \E_{\signal{s}}[\allocation(\signal{b}, \signal{s})]$  cannot be too large. The second lemma shows that if $\mechanism$ guarantees a good approximation ratio, then this sum must be large.  Together, we get a contradiction, which implies that $\mechanism$ cannot
% simultaneously have a good approximation ratio and be interim IR.

% We begin by stating Lemma~\ref{lemma:buyer-has-signal-impossibility-small-alloc}, then use it to derive the first part of the theorem. Finally, we provide the proof of the lemma. 

\begin{lemma}\label{lemma:buyer-has-signal-impossibility-small-alloc}
    If $\mechanism$ is BIC and interim IR then:
    $$   \sum_{\signal{b}=1}^m \E_{\signal{s}}[\allocation(\signal{b}, \signal{s})] \leq  2k(1+\E_{\signal{s}}[\buyervaluationsellerfunc{}(\signal{s})] ).
    $$
\end{lemma}

We start by computing the expected value of both agents. Note that since the  buyer is $\informedcoeffbuyer$-informed, we have $\E_{\signal{s}}[\buyervaluationsellerfunc{}(\signal{s})] = \frac{1-\informedcoeffbuyer}{\informedcoeffbuyer}\cdot\E_{\signal{b}}[\buyervaluationbuyerfunc{}(\signal{b})].$ Then
    %\begin{equation*}
    % \label{equality:impossibility-buyer-expected-value}
        %\E_{\signal{b},\signal{s}}[v_b(\signal{b}, \signal{s})] = \E_{\signal{s}}[\buyervaluationsellerfunc{}(\signal{s})] + \sum_{\signal{b}=1}^m \frac{1}{k^{\signal{b}}}\cdot k^{\signal{b}} = \frac{m}{\informedcoeffbuyer} \hspace{1cm} \text{and} \hspace{1cm} \E[v_s(\signal{b})] = \reversecoeff \cdot \mu + \sum_{\signal{b}=1}^m \frac{1}{k^{\signal{b}}}\cdot k^{\signal{b}-1} = \frac{m}{k}(\frac{1}{1-\informedcoeffseller}).
    %\end{equation*}
\begin{comment}
    \begin{align*}
    \E_{\signal{b},\signal{s}}[v_b(\signal{b}, \signal{s})] 
    &= \E_{\signal{s}}[\buyervaluationsellerfunc{}(\signal{s})] + \E_{\signal{b}}[\buyervaluationbuyerfunc{}(\signal{b})]
    = \left(1 + \frac{1-\informedcoeffbuyer}{\informedcoeffbuyer} \right) \E_{\signal{b}}[\buyervaluationbuyerfunc{}(\signal{b})] 
    = \frac{1}{\informedcoeffbuyer} \sum_{\signal{b}=1}^m \frac{1}{k^{\signal{b}}}\cdot k^{\signal{b}} 
    = \frac{m}{\informedcoeffbuyer} & \text{$\informedcoeffbuyer$-informed} \\
    \E[v_s(\signal{b})] &= \reversecoeff \cdot \mu + \sum_{\signal{b}=1}^m \frac{1}{k^{\signal{b}}}\cdot k^{\signal{b}-1} = \frac{\informedcoeffseller}{1-\informedcoeffseller} \cdot \frac{m}{k} + \frac{1}{k} \cdot m = \frac{m}{k}(\frac{1}{1-\informedcoeffseller}). & \gamma \cdot \mu = \frac{\informedcoeffseller}{1-\informedcoeffseller} \cdot \frac{m}{k}
\end{align*}
\end{comment}

\begin{align*}
    \E_{\signal{b},\signal{s}}[v_b(\signal{b}, \signal{s})] &= \left(1 + \frac{1-\informedcoeffbuyer}{\informedcoeffbuyer} \right) \E_{\signal{b}}[\buyervaluationbuyerfunc{}(\signal{b})] \quad \text{and} 
    % buyer
    & \E[v_s(\signal{b})] &= \reversecoeff \cdot \mu + \sum_{\signal{b}=1}^m \frac{1}{k^{\signal{b}}}\cdot k^{\signal{b}-1} & \informedcoeffbuyer\text{-informed} \\
    % new line
    &= \frac{1}{\informedcoeffbuyer} \sum_{\signal{b}=1}^m \frac{1}{k^{\signal{b}}}\cdot k^{\signal{b}} 
    & &= \frac{\informedcoeffseller}{1-\informedcoeffseller} \cdot \frac{m}{k} + \frac{1}{k} \cdot m & \gamma \cdot \mu  = \frac{\informedcoeffseller}{1-\informedcoeffseller} \cdot \frac{m}{k}  \\
    & = \frac{m}{\informedcoeffbuyer}
    & &= \frac{m}{k}(\frac{1}{1-\informedcoeffseller})
\end{align*}    

% \begin{lemma}\label{lemma:tec-bound-informed-part}
%     $$
%     \E_{\signal{s}}[\E_{\signal{b}}[\allocation(\signal{b}, \signal{s})\cdot\buyervaluationsellerfunc{}(\signal{s})]] \leq \E_{\signal{s}}[\sum_{\signal{b}=1}^m\frac{ \allocation(\signal{b}, \signal{s})}{m}\cdot \buyervaluationsellerfunc{}(\signal{s})]
%     $$
% \end{lemma}

% Let $A = \sum_{\signal{b}=1}^m \E_{\signal{s}}[\allocation(\signal{b}, \signal{s})]$, 
We compute the expected welfare of $\mechanism$ as the seller’s value plus the gains from trade in the event there is a trade:
    \begin{align*}
         ALG & = \E_{\signal{b}}[v_s(\signal{b})] + \sum_{\signal{b}=0}^m \pdfbuyersignal(\signal{b})\cdot \E_{\signal{s}}[\allocation(\signal{b}, \signal{s})\left( v_b(\signal{b}, \signal{s}) -v_s(\signal{b}) \right)]\\
         & = \left(\frac{m}{k}(\frac{1}{1-\informedcoeffseller})\right) + \underbrace{\pdfbuyersignal(0)\cdot \E_{\signal{s}}\left[\allocation(0, \signal{s}) \cdot (\buyervaluationsellerfunc{}(\signal{s})-\mu\cdot\reversecoeff)\right]}_{\signal{b}=0} 
         \\ & \hspace{1cm}
         + \sum_{\signal{b}=1}^m \frac{1}{k^{\signal{b}}}\cdot \E_{\signal{s}}\left[\allocation(\signal{b}, \signal{s})\left( \buyervaluationsellerfunc{}(\signal{s})+ k^{\signal{b}}- k^{\signal{b}-1}- \reversecoeff\cdot\mu \right)\right]\\
         & \leq \frac{m}{k}(\frac{1}{1-\informedcoeffseller}) + \underbrace{\E_{\signal{b}}[\E_{\signal{s}}[\allocation(\signal{b}, \signal{s})\cdot\buyervaluationsellerfunc{}(\signal{s})]]}_{\text{includes $\signal{b}=0$}} + (1-\frac{1}{k})\cdot\sum_{\signal{b}=1}^m \E_{\signal{s}}[\allocation(\signal{b}, \signal{s})] & - \mu \cdot \gamma \leq 0\\
         & \leq \frac{m}{k}(\frac{1}{1-\informedcoeffseller}) + m\cdot\frac{1-\informedcoeffbuyer}{\informedcoeffbuyer} + \sum_{\signal{b}=1}^m \E_{\signal{s}}[\allocation(\signal{b}, \signal{s})]  & \allocation(\signal{b},\signal{s})\leq 1, 1-\frac{1}{k} \leq 1,  \\
         %&\leq \frac{m}{k}(\frac{1}{1-\informedcoeffseller}) + m\cdot\frac{1-\informedcoeffbuyer}{\informedcoeffbuyer} + \sum_{\signal{b}=1}^m \E_{\signal{s}}[\allocation(\signal{b}, \signal{s})].
         & \leq \frac{m}{k}(\frac{1}{1-\informedcoeffseller}) + m\cdot\frac{1-\informedcoeffbuyer}{\informedcoeffbuyer} + 2k(1+m\cdot\frac{1-\informedcoeffbuyer}{\informedcoeffbuyer}). & \text{Lemma~\ref{lemma:buyer-has-signal-impossibility-small-alloc}} 
    \end{align*}
%By Lemma~\ref{lemma:buyer-has-signal-impossibility-small-alloc}, $ALG \leq \frac{m}{k}(\frac{1}{1-\informedcoeffseller}) + m\cdot\frac{1-\informedcoeffbuyer}{\informedcoeffbuyer} + 2k(1+m\cdot\frac{1-\informedcoeffbuyer}{\informedcoeffbuyer})$. 
Using that $OPT = \E[\max \{v_b(\signal{b}, \signal{s}),v_s(\signal{b})\}] \geq \E[v_b(\signal{b}, \signal{s})] = \frac{m}{\informedcoeffbuyer}$, we now bound the approximation ratio of $\mechanism$:
    \begin{align*}
            \frac{OPT}{ALG} \geq \frac{m\frac{1}{\informedcoeffbuyer}}{\frac{m}{k}\frac{1}{1-\informedcoeffseller} + m\cdot\frac{1-\informedcoeffbuyer}{\informedcoeffbuyer} + 2k(1+m\cdot\frac{1-\informedcoeffbuyer}{\informedcoeffbuyer})} = \frac{\frac{1}{\informedcoeffbuyer}}{\frac{1}{k}\frac{1}{1-\informedcoeffseller} + \frac{1-\informedcoeffbuyer}{\informedcoeffbuyer} + 2k\frac{1}{m}+2k\frac{1-\informedcoeffbuyer}{\informedcoeffbuyer}} 
    \end{align*}
\end{proof}
Before proving Lemma~\ref{lemma:buyer-has-signal-impossibility-small-alloc}, we state a technical lemma we use for the proof. 

\begin{lemma}\label{lemma:buyer-signal-tec-summation}
For every $\signal{s}$,
    $$
   \sum_{\signal{b}=1}^{m} \frac{\int_{1}^{\signal{b}} {\allocation(b,\signal{s})\cdot k^b\cdot\ln{k}} \d b}{k^{\signal{b}}} \geq \sum_{\signal{b}=1}^{m-1}\allocation(\signal{b}, \signal{s})(1-\frac{1}{k^{m-\signal{b}}}) .$$
\end{lemma}

% We provide the proof of Lemma \ref{lemma:buyer-signal-tec-summation} after the proof of Lemma \ref{lemma:buyer-has-signal-impossibility-small-alloc}.
We defer the proof of Lemma~\ref{lemma:buyer-signal-tec-summation} to Appendix~\ref{sec:app:lemma:aux-lemma}.

\begin{proof}[Proof of Lemma~\ref{lemma:buyer-has-signal-impossibility-small-alloc}]
By Myerson's lemma for the buyer (Lemma~\ref{lemma:buyer-payment-formula})\footnote{Without loss of generality, we assume that the buyer's support is in the interval $[0,2^m]$ and we let the values in the interval that are not in the original support arbitrary small probabilities so they won't affect the properties of the original instance.}, for every $\signal{b} \in [m]$, the payment identity gives that:
\begin{align}\label{ineq:buyer-signal-payment-ineq}
    \E_{\signal{s}}[(p(\signal{b}, \signal{s}))] & = \E_{\signal{s}}[p(0, \signal{s}) -\allocation(0, \signal{s})\cdot v_b(0, \signal{s}) + v_b(\signal{b}, \signal{s})\cdot \allocation(\signal{b}, \signal{s})  - \int_{0}^{\signal{b}} {\allocation(b,\signal{s}) \left. \frac{\partial v_b(\signal{b},\signal{s})}{\partial \signal{b}} \right|_{\signal{b} = b}} \d b] \nonumber\\
    & \leq \E_{\signal{s}}[p(0, \signal{s}) -\allocation(0, \signal{s})\cdot v_b(0, \signal{s}) + v_b(\signal{b}, \signal{s})\cdot \allocation(\signal{b}, \signal{s})  - \int_{1}^{\signal{b}} {\allocation(b,\signal{s})\cdot k^b\cdot\ln{k}} \d b] \nonumber\\
    & \leq \E_{\signal{s}}[v_b(\signal{b}, \signal{s})\cdot \allocation(\signal{b}, \signal{s})  - \int_{1}^{\signal{b}} {\allocation(b,\signal{s})\cdot k^b\cdot\ln{k}} \d b].
\end{align}
Where the last inequality uses that a buyer's utility at $(0, \signal{s})$ must be non-negative due to interim IR.
% Note that equivalently, $p(0, \signal{s}) \leq \allocation(0, \signal{s})\cdot v_b(0, \signal{s})$, which we use again as we now analyze the seller's expected utility. \kgc{We don't need the following. Utility doesn't need to be defined as a gain over not selling, or that can be a comment in prelims. IR is defined using the RHS in Def~\ref{def:interim-IR}.} %Since $\mechanism$ is interim IR for the seller, the following condition holds:
% \begin{align*}
%      \E_{\signal{b}}[ \E_{\signal{s}}[v_s(\signal{b})\cdot\left(1-\allocation(\signal{b}, \signal{s})\right ) + \payment(\signal{b}, \signal{s})]] \geq \E_{\signal{b}}[ \E_{\signal{s}}[v_s(\signal{b})]] \iff \E_{\signal{b}}[ \E_{\signal{s}}[-v_s(\signal{b})\cdot\allocation(\signal{b}, \signal{s}) + \payment(\signal{b}, \signal{s})]] \geq 0
% \end{align*}
% We now analyze the left-hand side of the second inequality.
% \kgc{Assuming the above gets deleted and seller utility is rewritten as $p-vx$ for ease, I'm rearranging below.} 
% \ttnote{ex-ante utility}
If the mechanism $\mechanism$ is interim IR, then the seller’s expected utility must be non-negative:
$\E_{\signal{b}}[ \E_{\signal{s}}[\payment(\signal{b}, \signal{s}) -v_s(\signal{b})\cdot\allocation(\signal{b}, \signal{s})]] \geq 0$. We leverage this constraint and further analyze the components of the seller’s
expected utility in this instance. Plugging in our instance for the seller’s expected utility, we get:
\begin{align*}
    \E_{\signal{b}} & [ \E_{\signal{s}} [\payment(\signal{b}, \signal{s}) -v_s(\signal{b})\cdot\allocation(\signal{b}, \signal{s})]]  \\
     & \leq  \pdfbuyersignal(0)\cdot \E_{\signal{s}}\left[\allocation(0, \signal{s}) \cdot \buyervaluationsellerfunc{}(\signal{s}) -\reversecoeff \cdot \mu\cdot \allocation(0,\signal{s})\right] \\
     & \hspace{2cm}+ \sum_{\signal{b}=1}^{m} \frac{1}{k^{\signal{b}}}\cdot \E_{\signal{s}}[\payment(\signal{b}, \signal{s}) -v_s(\signal{b})\cdot\allocation(\signal{b}, \signal{s})] & \text{buyer interim IR}\\
    & \leq  \pdfbuyersignal(0)\cdot \E_{\signal{s}}\left[\allocation(0, \signal{s}) \cdot \buyervaluationsellerfunc{}(\signal{s})\right] 
    \\ & \hspace{0.5cm} 
    + \sum_{\signal{b}=1}^{m} \frac{1}{k^{\signal{b}}}\cdot \E_{\signal{s}}\left[\allocation(\signal{b}, \signal{s})\left( k^{\signal{b}-1}(k-1) + \buyervaluationsellerfunc{}(\signal{s}) - \reversecoeff \cdot \mu \right) - \int_{1}^{\signal{b}} {\allocation(b,\signal{s})\cdot k^b\cdot\ln{k}} \d b\right] & \text{Inequality~\ref{ineq:buyer-signal-payment-ineq}}\\
    & \leq \E_{\signal{s}}[\buyervaluationsellerfunc{}(\signal{s})] + \sum_{\signal{b}=1}^{m} (1-\frac{1}{k})\cdot \E_{\signal{s}}[\allocation(\signal{b}, \signal{s})] -\E_{\signal{s}}\left[\sum_{\signal{b}=1}^{m} \frac{\int_{1}^{\signal{b}} {\allocation(b,\signal{s})\cdot k^b\cdot\ln{k}} \d b}{k^{\signal{b}}}\right] & \allocation(\signal{b},\signal{s})\leq1, -\gamma \cdot \mu \leq 0\\
    & \leq \E_{\signal{s}}[\buyervaluationsellerfunc{}(\signal{s})] +  \E_{\signal{s}}\left[\sum_{\signal{b}=1}^{m}(1-\frac{1}{k})\cdot\allocation(\signal{b}, \signal{s}) \right] - \E_{\signal{s}}\left[\sum_{\signal{b}=1}^{m-1}\allocation(\signal{b}, \signal{s})(1-\frac{1}{k^{m-\signal{b}}}) \right] & \text{Lemma~\ref{lemma:buyer-signal-tec-summation}}\\
    & \leq \E_{\signal{s}}[\buyervaluationsellerfunc{}(\signal{s})] + 1-\frac{1}{k} +\E_{\signal{s}}\left[\sum_{\signal{b}=1}^{m-1}\allocation(\signal{b}, \signal{s})\left(\frac{1}{k^{m-\signal{b}}}-\frac{1}{k}\right) \right] & \allocation(m, \signal{s}) \leq 1 \\
    &\leq \E_{\signal{s}}[\buyervaluationsellerfunc{}(\signal{s})] + 1 -\frac{1}{2k}\cdot\sum_{\signal{b}=1}^m \E_{\signal{s}}[\allocation(\signal{b}, \signal{s})].
\end{align*}%The first inequality is by interim IR for a buyer with $\signal{b}=0$. The second Inequality is by Inequality~\ref{ineq:buyer-signal-payment-ineq}. The forth inequality is due to Lemma~\ref{lemma:buyer-signal-tec-summation}. 
The last inequality holds because $\frac{1}{k^{m-{\signal{b}}}}-\frac{1}{k} \leq -\frac{1}{2k}$ for $1 \leq \signal{b} \leq m-2$ when $k\geq 2$ and $m\geq 3$. Additionally, $\frac{1}{k^{m-{\signal{b}}}}-\frac{1}{k} = 0$ for $\signal{b} = m-1$.

Thus, if $\mechanism$ is interim IR for the seller, then $\E_{\signal{s}}[\buyervaluationsellerfunc{}(\signal{s})] + 1 -\frac{1}{2k}\cdot\sum_{\signal{b}=1}^m \E_{\signal{s}}[\allocation(\signal{b}, \signal{s})] \geq 0.$ 
\end{proof}

\section{An Uninformed Seller: The $(0,\informedcoeffbuyer)$ Edge}\label{sec:seller-uninformed}

Recall that if the buyer is fully informed ($\informedcoeffbuyer =1$), no BIC and interim IR mechanism can guarantee a constant approximation ratio (Theorem~\ref{thm:buyer-informed-inpossibility}). In Subsection \ref{subsec-mechanism-0-beta}, we complement this result with a possibility: when the buyer is $\beta$-informed for $\beta \in [0,1)$, there exists a BIC and interim IR mechanism with an approximation ratio of $O(\frac{1}{1-\informedcoeffbuyer})$. That is, the less information the buyer has, the better the approximation ratio that we get. In Subsection \ref{subsec-impossibility-0-beta}, we show that, indeed, the optimal approximation ratio must improve as the value of $\informedcoeffbuyer$ decreases.

%In particular, note that as $\beta$ approaches $1$, the approximation ratio becomes unbounded, which aligns with the impossibility result for $\informedcoeffbuyer =1$. 

\subsection{A Mechanism for the $(0,\informedcoeffbuyer)$ Edge}\label{subsec-mechanism-0-beta}

\begin{theorem}\label{thm:seller-uninformed-possibility}
    Suppose that the seller is uninformed (i.e., $\informedcoeffseller = 0$), and the buyer is $\informedcoeffbuyer$-informed for $\informedcoeffbuyer \in [0,1)$. There exists a BIC and interim IR mechanism with an approximation ratio of $O(\frac{1}{1-\informedcoeffbuyer})$. 
\end{theorem}

% \kgnote{Would be good to add a quick primer on this sort of $\E[g]$ and $\mu$ and $\gamma$ notation in the prelims and how it's used to craft a give $(\alpha,\beta)$-information structure? Seems to be used many times throughout.}

\begin{proof}[Proof of Theorem~\ref{thm:seller-uninformed-possibility}]

By the statement, we have $v_s(\signal{b}, \signal{s}) = \sellervaluationbuyerfunc{}(\signal{b}), \, v_b(\signal{b}, \signal{s}) = \buyervaluationbuyerfunc{}(\signal{b}) + \buyervaluationsellerfunc{}(\signal{s})$.
Let  $\mu_s = \E_{\signal{b}}[\sellervaluationbuyerfunc{}(\signal{b})], \, \mu_1 = \E_{\signal{b}}[\buyervaluationbuyerfunc{}(\signal{b})],\, \mu_2 = \E_{\signal{s}}[\buyervaluationsellerfunc{}{}(\signal{s})]$. Since the buyer is $\informedcoeffbuyer$-informed for $\informedcoeffbuyer \in [0,1)$, we have that $\mu_1 = \mu_2 \cdot \frac{\informedcoeffbuyer}{(1-\informedcoeffbuyer)}$.  All together, that gives that 
$$\E[v_b(\signal{b},\signal{s})]= \mu_1 + \mu_2 \hspace{1cm} \text{and} \hspace{1cm} \E[v_s(\signal{b})]= \mu_s$$

We consider two cases. In the first case, where $\mu_s > \mu_2$ we do not trade the item, and the approximation ratio is at most:

$$\frac{OPT}{ALG} \leq \frac{\mu_s +\mu_1+\mu_2}{\mu_s} = 1+ \frac{\mu_2 \cdot \frac{\informedcoeffbuyer}{(1-\informedcoeffbuyer)}+\mu_2}{\mu_s} \leq 1 + \frac{\mu_2\cdot(\frac{1}{1-\informedcoeffbuyer})}{\mu_2} = 1+ \frac{1}{1-\informedcoeffbuyer},$$

where the first equality follows from $\mu_1 = \mu_2 \cdot \frac{\informedcoeffbuyer}{(1-\informedcoeffbuyer)}$ and the following inequality from $\mu_s > \mu_2$ case.

In the second case, $\mu_s \leq \mu_2$, we set a posted price of $p=\mu_2$. We claim that both agents always accepting the trade is an equilibrium. Given that the buyer always accepts the offer, the seller's best response is also to accept, as their expected value is at most $p$. Similarly, the buyer, knowing the seller will always accept, has an expected value of at least the price $p$. In this case, the approximation ratio of $p=\mu_2$ is at most:

$$\frac{OPT}{ALG} \leq \frac{\mu_s +\mu_1+\mu_2}{\mu_1+\mu_2} \leq  \frac{\mu_2\cdot\frac{2-\informedcoeffbuyer}{1-\informedcoeffbuyer}}{\mu_2 \cdot (\frac{1}{1-\informedcoeffbuyer})} = 2-\informedcoeffbuyer \leq  2.$$

Combining both cases, we conclude that the approximation ratio is $O(\frac{1}{1-\informedcoeffbuyer})$.

\end{proof}

\subsection{An Impossibility for the $(0,\informedcoeffbuyer)$ Edge}\label{subsec-impossibility-0-beta}

% \ttnote{Since we have been using $\gamma$ throughout the proofs for $\frac{1-\informedcoeffseller}{\informedcoeffseller}$, maybe change this parameter to something else? Or just use $\frac 1 \delta$}
\begin{theorem}\label{thm:imposs-uninformed-seller}
    Fix $\informedcoeffbuyer \geq 0.9$. There exists an information structure where the seller is fully uninformed and the buyer is $\informedcoeffbuyer$-informed, for which no BIC and interim IR mechanism can provide an approximation ratio better than $\Omega (\frac{1}{\sqrt{1-\informedcoeffbuyer}})$.
\end{theorem}

% \begin{proof}

To prove the theorem, we apply the first part of Theorem~\ref{thm:buyer-informed-inpossibility} with the following parameters; $\informedcoeffseller=0, \,\informedcoeffbuyer \geq 0.9,\, m=2k^2,$ and $k= \lceil \frac{1}{\sqrt{\lambda}}\rceil$, where $\lambda = \frac{1-\informedcoeffbuyer}{\informedcoeffbuyer}$. Then for the information structure $\mathcal{I}_{k,m, \informedcoeffseller, \informedcoeffbuyer}$ there is no BIC and interim IR mechanism with an approximation ratio better than:
$$
        \frac{\frac{1}{\informedcoeffbuyer}}{\frac{1}{k}\frac{1}{1-\informedcoeffseller} + \frac{1-\informedcoeffbuyer}{\informedcoeffbuyer} + 2k\frac{1}{m}+2k\frac{1-\informedcoeffbuyer}{\informedcoeffbuyer}} = \frac{1+ \lambda}{\frac{2}{k}+ \lambda+ 2k\lambda}  \geq \frac{1+\lambda}{2\sqrt{\lambda}+ \lambda + 2\cdot\lambda \left( \frac{1}{\sqrt{\lambda}} +1 \right)} = \frac{1+\lambda}{3\lambda + 4\sqrt{\lambda}} \geq \frac{1}{7\sqrt{\lambda}} \geq \frac{\frac{1}{7}\cdot\sqrt{0.9}}{\sqrt{1-\informedcoeffbuyer}}
$$

\section{An Uninformed Buyer: The $(\informedcoeffseller,0)$ Edge}\label{sec:uninformed-buyer}

% \kgc{Might be good in both to talk about how (1) refers to the whole square/interior, and then (2) is a corollary for the edge, using this language.}

Recall that if both players are fully uninformed then a (trivial) constant approximation mechanism exists. This section shows that when the seller is only slightly informed, then we can no longer construct constant approximation mechanisms.

\begin{theorem}\label{thm:impossibility-alpha-0-edge}
    For every $\informedcoeffseller > 0$, and every $c>1$, there exists an information structure where the buyer is totally uninformed and the seller is $\informedcoeffseller$-informed, and no BIC and interim IR mechanism provides an approximation ratio better than $c$.
\end{theorem}

% \begin{theorem}\label{thm:uninformed-buyer-impossibility}The following two statements hold:
% \begin{enumerate}
% \item For every $k, m \in \mathbb{N}$, there exists an information structure where the seller is $\informedcoeffseller$-informed and the buyer is $\informedcoeffbuyer$-informed such that $\informedcoeffseller > 0, \informedcoeffbuyer < 1, k\geq 4, m\geq 5$, and no BIC and interim IR mechanism can have an approximation ratio better than:
%     \begin{equation*}
% \frac{\informceoffbuyer+1}{\frac{1}{k}(1+ \informceoff)+ \informceoffbuyer  +\frac{4}{m}+\informceoffbuyer\cdot \frac{2}{k}+\frac{2}{k^2}\cdot\informceoff},
%     \end{equation*}
%     where $\informceoffbuyer =\frac{\informedcoeffbuyer}{1-\informedcoeffbuyer}$ and $ \sellerratioinformed = \frac{1-\informedcoeffseller}{\informedcoeffseller}$.
% \item As a corollary. for every $\informedcoeffseller > 0$, and every $c>1$, there exists an information structure where the buyer is totally uninformed and the seller is $\informedcoeffseller$-informed, and no BIC and interim IR mechanism provides an approximation ratio better than $c$.
% \end{enumerate}
% \end{theorem}
% \ttnote{I checked the substitutions, but its not obvious at all why the last inequality holds. Also (as stated later) I think there is an issue with the denominator here.}\asnote{I answered you there}

To prove the theorem, we apply the first part of Theorem~\ref{thm:uninformed-buyer-impossibility}. Fix $\informedcoeffseller >0$, and $c>1$. Let $\informedcoeffbuyer=0, \, \sellerratioinformed = \frac{1-\informedcoeffseller}{\informedcoeffseller}, \,\informceoffbuyer =\frac{\informedcoeffbuyer}{1-\informedcoeffbuyer} =0, \, k=\lceil 4c(3\sellerratioinformed +2) \rceil, \,  m= 4k$.
By the first part, there is an information structure $\mathcal{I}_{k,m, \informedcoeffseller, \informedcoeffbuyer}$ where the buyer is fully uninformed and the seller is $\informedcoeffseller$-informed, and no BIC and interim IR mechanism has an approximation ratio better than:
    \begin{equation*}
\frac{\informceoffbuyer+1}{\frac{1}{k}(1+ \informceoff)+ \informceoffbuyer  +\frac{4}{m}+\informceoffbuyer\cdot \frac{2}{k}+\frac{2}{k}\frac{1}{k}\cdot\informceoff}= \frac{1}{\frac{1}{k}(1+ \informceoff)+ \frac{1}{k}+\frac{2}{k}\frac{1}{k}\cdot\informceoff}\geq \frac{1}{\frac{1}{k}(2+ \informceoff +2\informceoff)} > c.
    \end{equation*}

\section{The Interior of the Square}\label{sec-interior}

Up until now, we only discussed information structures on the edges of Figure \ref{fig:intro-information-rectangle}, i.e., we had at least one player that was informed, or at least one player that was uninformed. In this section we consider information structures where both players are partially but not fully informed. That is, in this section, we study some areas in the interior of the square depicted in Figure \ref{fig:intro-information-rectangle}. The main take-home message of this section is that the positive results we achieved in the previous section are very sensitive. They require the seller to be either fully informed or fully uninformed. If the seller gains or loses even a tiny bit of information, the possible approximation ratio deteriorates significantly.

\begin{proposition}\label{prop:left-corner-lb}
    For every $\informedcoeffseller > 0$ and $\informedcoeffbuyer < 1$, there exists an $(\informedcoeffseller,\informedcoeffbuyer)$-information structure where no BIC and interim IR mechanism can provide an approximation ratio better than $\frac{1}{2\informedcoeffbuyer}$.
\end{proposition}

Fix some $\beta>0$ and recall that for $(0,\beta)$-information structures we get an approximation ratio of $\frac 1 {1-\beta}$ (Theorem~\ref{thm:seller-uninformed-possibility}). The proposition shows that if the buyer has a very small amount of information, then the approximation ratio deteriorates significantly. That is, even if $\alpha>0$ is very small, the possible approximation ratio for $(\alpha,\beta)$-information structures is $O(\frac 1 \beta)$, which goes to $\infty$ as $\beta$ goes to $0$. Returning to Figure \ref{fig:intro-information-rectangle}, we get that none of the points in the bottom-left corner admit a good approximation ratio except for those that are on the edge $(0,\beta)$.

% Before proving the corollaries, note that for the second corollary we have $\informedcoeffseller(\informedcoeffseller^2 -3\informedcoeffseller +3) < 1$ for $\informedcoeffseller < 1$ \footnote{The function $f(x) = x^3-3x^2+3x$, is strictly increasing in $(0,1)$ as it's derivative is strictly positive for this range; $f'(x) = 3(x^2-2x+1)$. Hence, the maximum value of $f(x)$ in the range $(0,1)$ is smaller than at $f(1) = 1$.}.

\begin{proof}[Proof of Proposition~\ref{prop:left-corner-lb}]
    We apply the first part of Theorem~\ref{thm:uninformed-buyer-impossibility} with the parameters $k=\lceil \frac{3\informceoff+4\informceoffbuyer+5}{\informceoffbuyer} \rceil$ and $m=k^2$, where  $\informceoffbuyer =\frac{\informedcoeffbuyer}{1-\informedcoeffbuyer}$ and $ \sellerratioinformed = \frac{1-\informedcoeffseller}{\informedcoeffseller}$. This implies the existence of an information structure $\mathcal{I}_{k,m,\informedcoeffseller, \informedcoeffbuyer}$ for which no BIC and interim IR mechanism can have an approximation ratio better than:

\begin{equation*}
\frac{\informceoffbuyer+1}{\frac{1}{k}(1+ \informceoff)+ \informceoffbuyer  +\frac{4}{m}+\informceoffbuyer\cdot \frac{2}{k}+\frac{2}{k^2}\cdot\informceoff}\geq \frac{\informceoffbuyer+1}{\frac{1}{k}(5+ 3\informceoff +2\informceoffbuyer) + \informceoffbuyer} \geq \frac{\informceoffbuyer +1}{2\informceoffbuyer} = \frac{1}{2\informedcoeffbuyer}.
\end{equation*}
\end{proof}

The next proposition has similar qualitative applications (but with somewhat worse parameters). The approximation ratio that we get for $(1,\beta)$-information structures is $O(\frac 1 \beta)$ (Theorem~\ref{thm:possibility-seller-informed}). The proposition shows that if the seller is only somewhat less informed, then the approximation ratio significantly deteriorates: for values of $\alpha$ that approach $1$ but are still strictly smaller to $1$, the approximation ratio approaches infinity. We get that the only points in the top right corner of Figure \ref{fig:intro-information-rectangle} that admit a good approximation are those on the edge $(1,\beta)$.

\begin{proposition}\label{prop:right-corner-lb}
    For every $\informedcoeffseller \in (0.9,1)$ and $ \informedcoeffbuyer \in [1-(1-\informedcoeffseller)^3,1) $, there exists an $(\informedcoeffseller,\informedcoeffbuyer)$-information structure where no BIC and interim IR mechanism can provide an approximation ratio better than $\frac{0.15}{1-\informedcoeffseller}$.
 \end{proposition}

\begin{proof}[Proof of Proposition~\ref{prop:right-corner-lb}]
We apply the first part of Theorem~\ref{thm:buyer-informed-inpossibility} with the parameters $k=\lceil \left(\frac{1}{1-\informedcoeffseller}\right)^2 \rceil$ and $m= 20k\cdot \lceil \frac{1}{1-\informedcoeffseller} \rceil$. This implies the existence of an information structure $\mathcal{I}_{k,m,\informedcoeffseller, \informedcoeffbuyer}$ for which no BIC and interim IR mechanism can have an approximation ratio better than:

     \begin{equation*}
        \frac{\frac{1}{\informedcoeffbuyer}}{\frac{1}{k}\frac{1}{1-\informedcoeffseller} + \frac{1-\informedcoeffbuyer}{\informedcoeffbuyer} + 2k\frac{1}{m}+2k\frac{1-\informedcoeffbuyer}{\informedcoeffbuyer}} \geq \frac{1}{(1-\informedcoeffseller) + \frac{1-\informedcoeffseller}{\informedcoeffbuyer} + \frac{(1-\informedcoeffseller)}{10}+2(\left(\frac{1}{1-\informedcoeffseller}\right)^2 +1)\frac{1-\informedcoeffbuyer}{\informedcoeffbuyer}} \geq \frac{1}{(1-\informedcoeffseller)(1+\frac{1}{0.9}+\frac{1}{10} + \frac{4}{0.9})} > \frac{0.15}{1-\informedcoeffseller}
    \end{equation*}
where the inequalities hold since $\informedcoeffbuyer \geq  1-(1-\informedcoeffseller)^3$, and so $(1-\informedcoeffbuyer) \leq (1-\informedcoeffseller)^3$, and $\informedcoeffbuyer \geq \informedcoeffseller \geq 0.9$.
\end{proof}

% \begin{itemize}
%     \item Case 1: $\informedcoeffseller>0$. Impossibility of $\frac{1}{\beta}$, by Theorem~\ref{thm:uninformed-buyer-impossibility}, where $m=4k, k=(2+3\informceoff+2\informceoffbuyer)\cdot c$, the lower bound approaches $\frac{1}{\informedcoeffbuyer}$ when $c$ approaches $\infty$. 
%     \asnote{demonstrates that the possibility of section 5 cannot be extended beyond alpha=0.}
%     % \item Case 2: we can get 
%     % $\frac{1-\informedcoeffseller}{2(1-\informedcoeffbuyer)}$, that is meaningful for large values of $\beta$. Need to optimize the analysis of 4.1
%     \item For large $\informedcoeffbuyer, \informedcoeffseller$, we can apply 4.1 with $k=\frac{\informedcoeffseller\cdot c}{(1-\informedcoeffseller)}$ and get approximately:
%     $$
%     \frac{\frac{1}{\informedcoeffbuyer}}{\frac{1}{c\informedcoeffseller}+\frac{1-\informedcoeffbuyer}{\informedcoeffbuyer}(1+\frac{c\informedcoeffseller}{1-\informedcoeffseller})}.
%     $$
% \end{itemize}

    % \asnote{For large $\alpha$ we get $\frac{1}{1-\informedcoeffbuyer}$ from 4.1}

%% file: polynomials.tex
\section{Information Structures with Polynomial Functions}\label{sec:poly-function}

% In this section, we consider the case where both agents have separable linear functions in both agents' signals, and both signals are independently sampled from a uniform distribution on $[0,1]$. Let:

In this section, we consider the case where the signals are independently sampled from a uniform distribution on $[0,1]$, and that both agents' valuations are polynomials functions of both signals (i.e., $\buyervaluationbuyerfunc{}(\signal{b}), \buyervaluationsellerfunc{}(\signal{s}), \sellervaluationsellerfunc{} (\signal{s})$, and  $\sellervaluationbuyerfunc{}(\signal{b})$ are all polynomials):

\begin{subequations}\label{equ:valuations-linear}
\begin{equation}\label{equ:linear-functions-buyer-val}
    v_b(\signal{b}, \signal{s}) = \sum_{i=1}^{k} a_i\cdot\signal{b}^i + \sum_{i=1}^k b_i \cdot \signal{s}^i + \constantbuyerval
\end{equation}

\begin{equation}\label{equ:linear-functions-seller-val}
    v_s(\signal{b}, \signal{s}) = \sum_{i=1}^{k} c_i\cdot\signal{s}^i + \sum_{i=1}^k d_i \cdot \signal{b}^i + \constantsellerval
\end{equation}
\end{subequations}

We show that there is a mechanism that achieves an approximation ratio of $O(k^2)$ (Subsection \ref{subsec-polynomial-mechanism}), and that every mechanism achieves an approximation ratio of $\Omega(k)$ (Subsection \ref{subsec-polynomial-impossiblity}). In particular, for linear functions, we can guarantee a constant approximation ratio.

\subsection{An $O(k^2)$ Approximation Mechanism}\label{subsec-polynomial-mechanism}

\begin{theorem}\label{thm:polynomial-degree-squared}
    Suppose that $v_s,v_b$ are polynomials of maximum degree $k$, and that the signals are independently drawn from a uniform distribution over $[0,1]$. Then, there exists a BIC and interim IR mechanism that guarantees an approximation ratio of $O(k^2)$. In particular, when $v_s,v_b$ are linear functions, the approximation ratio is constant.
\end{theorem}

\begin{proof}
We divide the analysis into two cases. In the first case, the expected value of the seller is large enough in comparison to the expected value of the buyer, so the mechanism that never trades the item achieves a good approximation ratio. In the second case, the expected value of the seller is small relative to the buyer's, and so the item must sometimes be traded in order to achieve a good approximation ratio.

Let $\applinearproof = (k+1)^2$. We show that in both cases the approximation ratio is at most $ 1+ \applinearproof$. In the first case, we have that $\mathbb E[v_s] \geq \frac{\mathbb E[v_b]}{\applinearproof}$. Since the optimal welfare is at most $\mathbb E[v_s] + \mathbb E[v_b]$, the approximation ratio of the mechanism that never trades the item is at most $\frac{\mathbb E[v_s] + \mathbb E[v_b]}{\mathbb E[v_s]} \leq 1+ \applinearproof$.

For the remainder of the proof we assume that $\mathbb E[v_s] < \frac{\mathbb E[v_b]}{\applinearproof}$. By Equation~\eqref{equ:valuations-linear} and the assumption that the signals are drawn from a uniform distribution over $[0,1]$, this inequality can be expressed as:

\begin{align}\label{equ:linear-case-of-trade}
    \sum_{i=1}^k \frac{c_i}{i+1} + \sum_{i=1}^k \frac{d_i}{i+1} + \constantsellerval &< \frac{\sum_{i=1}^k \frac{a_i}{i+1} + \sum_{i=1}^k \frac{b_i}{i+1} + \constantbuyerval}{\applinearproof} \nonumber
    \\ \iff \nonumber\\
    \sum_{i=1}^k c_i\frac{(k+1)}{i+1} + \sum_{i=1}^k d_i\frac{(k+1)}{i+1} + (k+1)\cdot \constantsellerval & < \frac{\sum_{i=1}^k a_i\frac{(k+1)}{i+1} + \sum_{i=1}^k b_i\frac{(k+1)}{i+1} + (k+1)\cdot \constantbuyerval}{\applinearproof}
\end{align}

\begin{lemma}\label{lemma:equ-strategies-thresholds}
    For a posted price $p$, the players' equilibrium strategies are thresholds $( \thsbuyer,\thsseller)$.Specifically, the seller agrees to trade if and only if $\signal{s} \leq \thsseller$, and the buyer agrees to a trade if and only if $\signal{b} \geq \thsbuyer$.
\end{lemma}

Consider a posted price mechanism with price $p$, by Lemma~\ref{lemma:equ-strategies-thresholds}, the agents follow threshold strategies $(\thsbuyer,\thsseller)$, where, $\thsbuyer$ is the buyer's threshold and $\thsseller$ is the seller's threshold. 
In equilibrium, these strategies are best responses.  Then if the seller's threshold is $\thsseller \in [0,1]$, then the buyer will choose a threshold $\thsbuyer \in [0,1]$ such that Inequality~(\ref{inequ:buyer-stratgy-non-negative}) is satisfied. Similarly, if the buyer's threshold is 
$\thsbuyer \in [0,1]$, the seller will choose a threshold 
$\thsseller \in [0,1]$ such that Inequality~(\ref{inequ:seller-stratgy-non-negative}) is satisfied.

\begin{subequations}\label{ineq:ths-non-negative-profit}
\begin{equation}\label{inequ:buyer-stratgy-non-negative}
    \forall \signal{b} \in [0,1] \text{ s.t. } \signal{b}\geq \thsbuyer: \qquad \mathbb E_{\signal{s}}[v_b(\signal{b}, \signal{s})| \signal{s} \leq \thsseller] \geq p \, \iff \,  
    \sum_{i=1}^k a_i\cdot{\signal{b}}^i + \sum_{i=1}^k \frac{b_i\cdot {\thsseller}^i}{i+1} +\constantbuyerval \geq p
\end{equation}

\begin{equation}\label{inequ:seller-stratgy-non-negative}
   \forall \signal{s} \in [0,1] \text{ s.t. } \signal{s}\leq \thsseller: \qquad  \mathbb E_{\signal{b}}[v_s(\signal{s}, \signal{b})| \signal{b} \geq \thsbuyer] \leq p \, \iff \,  
    \sum_{i=1}^k c_i\cdot{\signal{s}}^i + \sum_{i=1}^k d_i\cdot\frac{ \left(1-{\thsbuyer}^{i+1}\right)}{(i+1)(1-\thsbuyer)} +\constantsellerval \leq p
\end{equation}
\end{subequations}

Recall that we now assume that the expected value of the buyer is much larger than the expected value of the seller (i.e., Inequality~\eqref{equ:linear-case-of-trade} holds). 
We now divide the analysis into two cases, based on which of the terms dominates the expected value of the buyer: (1) $\frac{a_i}{i+1}$ for some $i \in [k]$ or (2) $\sum_{i=1}^k \frac{b_i}{i+1} + \constantbuyerval$. In all cases, we analyze the same posted price mechanism with price $p=\frac{\mathbb E[v_b]}{k+1}$.

In essence, we are dividing into the case where the seller's signal has a greater influence on the buyer's value, or the buyer's signal is a $k$-approximation.  Depending on the case, we can then use this to bound the buyer's valuation for each of $k$ terms in the $a_i$ sum and the second half of the sum:
\begin{align}
    \text{Using case (1)} &&\sum_{i=1}^k\frac{b_i}{i+1} + \constantbuyerval &= \frac{k+1}{k+1} \cdot \left( \sum_{i=1}^k\frac{b_i}{i+1} + \constantbuyerval \right) &\geq \frac{1}{k+1} \left( \sum_{i=1}^k\frac{a_i}{i+1} + \sum_{i=1}^k \frac{b_i}{i+1} + \constantbuyerval \right) \label{eq:case1kterms} \\ %\label{eq:case1kterms}
    \text{Using case (2)} && \max \{\frac{a_i}{i+1} | \, i \in [k]\} &= \frac{k+1}{k+1} \cdot \max \{\frac{a_i}{i+1} | \, i \in [k]\} &>  \frac{1}{k+1} \left( \sum_{i=1}^k\frac{a_i}{i+1} + \sum_{i=1}^k \frac{b_i}{i+1} + \constantbuyerval \right) \label{eq:case2kterms}
\end{align}

\begin{enumerate}
    \item $\max \{\frac{a_i}{i+1} | \, i \in [k]\} \leq  \sum_{i=1}^k \frac{b_i}{i+1} + \constantbuyerval$. In this case, we show that both players accept a trade, resulting in an approximation ratio of at least $1+\frac{1}{\applinearproof}$. First, observe that for $\applinearproof = (k+1)^2$, the seller's maximal value is smaller than the price $p$:
    \begin{align*}
        v_s(x_s,x_b) & \le  v_s(1,1) = \sum_{i=1}^k c_i +\sum_{i=1}^k d_i +\constantsellerval & x_s \leq 1 \\ 
        & \le   (k+1)\cdot\left(\sum_{i=1}^k \frac{c_i}{i+1} +\sum_{i=1}^k \frac{d_i}{i+1} +\constantsellerval\right) & i \leq k \\
        &\leq (k+1)\frac{\mathbb E[v_b]}{\applinearproof} = \frac{\mathbb E[v_b]}{(k+1)} = p.    & \text{by Eq.~\eqref{equ:linear-case-of-trade} and $\gamma = (k+1)^2$}
    \end{align*}
    %Where the fourth inequality follows Eq.~\eqref{equ:linear-case-of-trade}.
    
    Hence, the seller should accept a trade regardless of the buyer's strategy. The buyer's expected value, when their signal is $\signal{b}$, given that the seller accepts the offer regardless of their signal, is equal to 
    %$$\sum_{i=1}^k a_i\cdot{\signal{b}}^i + \sum_{i=1}^k\frac{b_i}{i+1} + \constantbuyerval \geq 0 +  \frac{1}{k+1}\left(\sum_{i=1}^k\frac{b_i}{i+1} + \constantbuyerval + \sum_{i=1}^k\frac{a_i}{i+1} \right) = \frac{\mathbb E[v_b]}{k+1},$$

    \begin{align*}
    \sum_{i=1}^k a_i\cdot{\signal{b}}^i + \sum_{i=1}^k\frac{b_i}{i+1} + \constantbuyerval & \geq 0 + \sum_{i=1}^k\frac{b_i}{i+1} + \constantbuyerval & x_b \geq 0 \\    
    %&= \frac{k+1}{k+1} \left(\sum_{i=1}^k\frac{b_i}{i+1} + \constantbuyerval \right) \\
    &\geq \frac{1}{k+1}\left(\sum_{i=1}^k\frac{b_i}{i+1} + \constantbuyerval + \sum_{i=1}^k\frac{a_i}{i+1} \right) & \text{by Eq. \eqref{eq:case1kterms}} \\
    %\text{By case (1) for each of $k$ terms}\\
    &= \frac{\mathbb E[v_b]}{k+1} = p.
    \end{align*}
    
    %which is at least $p$. The inequality follows because the buyer's signal is always at least 0, and each of the $k$ terms $\frac{a_i}{i+1}$ is bounded by $\sum_{i=1}^k \frac{b_i}{i+1} + \constantbuyerval$ by the case we're in. 
    Thus, both players accept a trade in equilibrium. 
    \item $\max \{\frac{a_i}{i+1} | \, i \in [k]\} >  \sum_{i=1}^k \frac{b_i}{i+1} + \constantbuyerval$. Let $j =\arg \max \{ \frac{a_i}{i+1} | \, i \in [k] \}$.
     We show that, in this case, the seller's best response is to always accept a trade regardless of the value of the buyer's threshold. We also show that for every threshold strategy of the seller, the buyer's threshold is at most $\left(\frac{1}{j+1}\right)^{\frac{1}{j}}$. I.e.,  the buyer accepts a trade when their signal is $\signal{b} \geq\left(\frac{1}{j+1}\right)^{\frac{1}{j}}$ and potentially even when their signal is lower.

We use monotonicity to lower-bound the buyer's value using just the portion that their own signal contributes to and a single term of the polynomial:
\begin{equation}\label{eq:buyer1termLB}
    v_b(\signal{b}, \signal{s}) 
    \quad = \quad  a_j\cdot \signal{b}^j + \underbrace{\sum_{i=1, i \neq j}^{k} a_i\cdot\signal{b}^i + \sum_{i=1}^k b_i\cdot\signal{s}^i + \constantbuyerval}_{\geq 0} \quad \geq \quad  a_j \cdot \signal{b}^j .
    \end{equation}
     
    Before proving that these are the thresholds in equilibrium, we show that if the players play these strategies, we get a good approximation ratio.
    The welfare of every mechanism that trades the item when the buyer's signal is at least $\left(\frac{1}{j+1}\right)^{\frac{1}{j}}$ is at least:
    \begin{align*}
    \mathbb E[v_b(\signal{b}, \signal{s})&\cdot \indicator_{\left[\signal{b} 
    \geq \left(\frac{1}{j+1}\right)^{\frac{1}{j}}\right]}] 
    %= \int_{\left(\frac{1}{j+1}\right)^{\frac{1}{j}}}^{1} \left(a_j\cdot x^j + \underbrace{\sum_{i=1, i \neq j}^{k} a_i\cdot\signal{}^i + \sum_{i=1}^k \frac{b_i}{i+1} + \constantbuyerval}_{\geq 0}\right) \d x \\
     \geq \int_{\left(\frac{1}{j+1}\right)^{\frac{1}{j}}}^{1} a_j\cdot x^j \d x  & \text{by Eq. \eqref{eq:buyer1termLB}} \\
    &= \frac{a_j}{j+1} \cdot x^{j+1} \bigg |_{\left(\frac{1}{j+1}\right)^{\frac{1}{j}}}^{1} 
    =  \frac{a_j}{j+1} \cdot \left(1- \left( \frac{1}{j+1}\right)^{\frac{j+1}{j}}\right) \\
    & \geq\frac{a_j}{j+1}\cdot\frac{3}{4} & \text{$\left( \frac{1}{j+1}\right)^{\frac{j+1}{j}} \leq 0.25$ for $j \geq 1$}\\
    &\geq \frac{3}{4} \left( \frac{\sum_{i=1}^k\frac{a_i}{i+1} +\sum_{i=1}^k \frac{b_i}{i+1} + \constantbuyerval}{k+1}\right) = \frac{3}{4}\left(\frac{\mathbb E[v_b]}{k+1}\right) & \text{by Eq. \eqref{eq:case2kterms}} %& \text{by def. of $j$ for $k$ terms and case (2) for 1 term}
    \end{align*}
    %where the last inequality follows from creating $k+1$ terms, applying the definition of $j$ to the $k$ terms of $\frac{a_i}{i+1}$ and then applying case (2) to $k+1$st $a_j$ term.
    %where the second inequality is since $\left( \frac{1}{j+1}\right)^{\frac{j+1}{j}} \leq 0.25$ for $j \geq 1$. 
    Therefore, the approximation ratio in this case is at most:
    $$
    \frac{\mathbb E[v_s] + \mathbb E[v_b]}{\frac{3}{4(k+1)}\mathbb E[v_b]} \leq \frac{4(k+1)}{3}+ \frac{\frac{1}{\applinearproof}\mathbb E[v_b]}{\frac{3}{4(k+1)}\mathbb E[v_b]} = \frac{4(k+1)}{3}(1+\frac{1}{\applinearproof})
    $$
    Next, we prove that for the posted price $p=\frac{\mathbb E[v_b]}{k+1}$, the seller always accepts and the buyer's threshold satisfies $\thsbuyer \geq \left(\frac{1}{j+1}\right)^{\frac{1}{j}}$ in equilibrium. A similar argument to the one used in the first case shows that the maximal value of the seller (when both signals are $1$) is smaller than or equal to the price, and so the seller's best response, regardless of the buyer's threshold, is to always accept the offer.
    By Eq. \eqref{eq:buyer1termLB}, the value of a buyer with signal $\signal{b} \geq  \left(\frac{1}{j+1}\right)^{\frac{1}{j}}$ is at least
    % \begin{multline*}
    %     \sum_{i=1}^k a_i\cdot{\signal{b}}^i + \sum_{i=1}^k \frac{b_i}{i+1}   +\constantbuyerval \geq \sum_{i=1}^k a_i\cdot{\left(\frac{1}{j+1}\right)^{\frac{i}{j}}} + \sum_{i=1}^k \frac{b_i}{i+1}   +\constantbuyerval \geq a_j\cdot{\left(\frac{1}{j+1}\right)} + \sum_{i=1}^k \frac{b_i}{i+1}   +\constantbuyerval \\ 
    %     \geq \frac{1}{k}\left(\sum_{i=1}^k \frac{a_i}{i+1}\right) + \sum_{i=1}^k \frac{b_i}{i+1}   +\constantbuyerval,
    % \end{multline*}
    \begin{align*}
        %\sum_{i=1}^k a_i\cdot{\signal{b}}^i  \geq \sum_{i=1}^k a_i\cdot{\left(\frac{1}{j+1}\right)^{\frac{i}{j}}}  \geq 
        a_j\cdot{\signal{b}}^j \geq 
        a_j\cdot{\left(\frac{1}{j+1}\right)} 
        \geq \frac{1}{k+1}\left(\sum_{i=1}^k \frac{a_i}{i+1} + \sum_{i=1}^k \frac{b_i}{i+1}   +\constantbuyerval\right)
    \end{align*}
    % $\coeffbuyervalbuyersignal\cdot\signal{b} + \coeffbuyervalsellersignal\frac{1}{2} + \constantbuyerval$. For $\signal{b} \geq \frac{1}{2}$,    
    Note that this value is at least the price $p$.
\end{enumerate}

Thus, in each case, the approximation ratio is at most $\max \{ 1+\applinearproof, \,\frac{4(k+1)}{3}\left(1+\frac{1}{\applinearproof}\right) \}$ for $\applinearproof = (k+1)^2$. The approximation ratio is, therefore, at most $1+(k+1)^2$.
\end{proof}

\begin{proof}[Proof of Lemma~\ref{lemma:equ-strategies-thresholds}]
Consider a posted price mechanism with price $p$. Note that if the buyer accepts a trade at price $p$ when their signal is $\thsbuyer$, they will also accept the trade for any higher signal $\signal{b} \geq \thsbuyer$, since the buyer's valuation function is monotonically non-decreasing in their signal.

Similarly, if the seller accepts a trade at price $p$ when their signal is $\thsseller$, they will also agree to trade for any lower signal $\signal{s} \leq \thsseller$. Therefore, an equilibrium strategy profile for the two agents can be described by two thresholds, $(\thsbuyer,\thsseller) \in [0,1]^2$: one threshold for the buyer and one threshold for the seller. A seller accepts a trade if and only if their value is at most $\thsseller$, and the buyer accepts a trade if and only if their value is at least $\thsbuyer$.
\end{proof}

% Note that if the buyer accepts a trade at price $p$ when their signal is $\thsbuyer$, they will also accept the trade for any signal $\signal{b} \geq \thsbuyer$, since the buyer's valuation function is monotonically non-decreasing in their signal.
%Similarly, if the seller accepts a trade at price $p$ when their signal is $\thsseller$, they will also agree to trade for any signal $\signal{s} \leq \thsseller$. Therefore, an equilibrium strategy profile for the two agents can be described by two thresholds, $(\thsbuyer,\thsseller) \in [0,1]^2$: one threshold for the buyer and one threshold for the seller. A seller accepts a trade if and only if their value is at most $\thsseller$, and the buyer accepts a trade if and only if their value is at least $\thsbuyer$.

% Moreover, we will see that no solution exists even when $\thsseller =1$, or $\thsbuyer =0$. Thus, for any price $p$, the only equilibrium strategy is to not trade at all.  

% \subsection{A Lower Bound of $\Omega(k)$}
% \begin{theorem}\label{thm:(1,0)-impossibility}
%     %For $\informedcoeffseller = 0$, and 
%     For every $c>1$, there exists a $(0,1)$-information structure such that %where the buyer is totally uninformed and the seller is $1$-informed, and 
%     no Bayesian incentive-compatible and interim individually rational mechanism can achieve an approximation ratio better than $c$.
% \end{theorem}

\subsection{An $\Omega(k)$ Impossibility}\label{subsec-polynomial-impossiblity}

\begin{theorem}\label{thm:polynomial-lowerbound}
    For every $k\in \mathbb{N}$, there exist polynomials $v_b,v_s$ of degree $k$ such that no BIC and interim IR mechanism can achieve an approximation ratio better than $k$.
\end{theorem}

\begin{proof}[Proof of Theorem~\ref{thm:polynomial-lowerbound}]

We prove our impossibility for the following construction $\mathcal{I}_{k}$.  % information structure $\mathcal{I}_{c}$ that we construct now.
Let the seller's signal $\signal{s}$ be drawn uniformly from $[0,1]$. Consider the following valuations:

$$
v_s(\signal{s})  = \sellervaluationsellerfunc{}(\signal{s}) = \signal{s}^k \hspace{2cm} \text{and} \hspace{2cm} 
v_b(\signal{s})  = \buyervaluationsellerfunc{}(\signal{s}) = k \signal{s}^k .
$$

Since the buyer does not have a signal in $\mathcal{I}_{k}$, the seller payment formula ~\eqref{lemma:seller-payment-formula} simplifies to the following expression for $\signal{s} \in [0,1]$, where $v'_s(z) = \frac{\partial v_s(\signal{s})}{\partial\signal{s}} |_{\signal{s} = z}$:

\begin{equation}\label{equ:(1,0)-payment-BIC}
\payment(\signal{s}) = \payment(0) + v_s (\signal{s}) \cdot \allocation(\signal{s}) - v_s(0) \cdot \allocation(0) - \int_{0}^{\signal{s}} {\allocation(z) v'_s(z)} \, dz.
\end{equation}

Combining the payment formula ~\eqref{equ:(1,0)-payment-BIC} with the seller interim IR condition, we get for $\signal{s} \in [0,1]$ that:

\begin{equation}\label{ineq:(1,0)-IR-bound}
\payment(\signal{s}) - v_s({\signal{s}}) \cdot \allocation(\signal{s}) \geq 0 \iff \payment(0) - v_s({0}) \cdot \allocation(0) - \int_{0}^{\signal{s}} {\allocation(z) v'_s(z)} \, dz  \geq 0 .
\end{equation}

We now consider the buyer's expected utility:
\begin{align}
    & \mathbb{E}_{\signal{s}}{[\allocation(\signal{s})v_b(\signal{s}) - \payment(\signal{s})]} = \int_{0}^{1}\left[ \allocation(\signal{s}) v_b({\signal{s}})  - \payment(\signal{s}) \right] \, d\signal{s} \nonumber\\
    & = \int_{0}^{1}\left[ \allocation(\signal{s}) \left(v_b({\signal{s}}) - v_s(\signal{s}) \right)  - \payment(0) + v_s(0) \cdot \allocation(0) + \int_{0}^{\signal{s}} {\allocation(z) v'_s(z)} \, dz \right] \, d\signal{s} \quad \tag*{Seller payment formula ~\eqref{equ:(1,0)-payment-BIC}}\nonumber\\
    & \leq \int_{0}^{1}\left[ \allocation(\signal{s}) \left(v_b({\signal{s}}) - v_s(\signal{s}) \right) - \int_{0}^1 {\allocation(z) v'_s(z)} \, dz  + \int_{0}^{\signal{s}} {\allocation(z) v'_s(z)} \, dz \right] \, d\signal{s} \tag*{Seller interim IR ~\eqref{ineq:(1,0)-IR-bound} for \text{$\signal{s} = 1$}} \nonumber\\
    & = \int_{0}^{1}\allocation(\signal{s}) \left(v_b({\signal{s}}) - v_s(\signal{s}) \right) \, d\signal{s} - \int_{0}^{1} \left[ \int_{0}^1{\allocation(z) v'_s(z)} \, dz \right ] d\signal{s} +  \int_{0}^{1} \left[\int_{0}^{\signal{s}} {\allocation(z) v'_s(z)} \, dz \right] \, d\signal{s} \nonumber\\
    & = \int_{0}^{1}\allocation(\signal{s}) \left(v_b({\signal{s}}) - v_s(\signal{s}) \right) \, d\signal{s} - \int_{0}^1{\allocation(z) v'_s(z)} \, dz \cdot \int_{0}^{1} d\signal{s} +  \int_{0}^{1} \left[\int_{z}^1 {\allocation(z) v'_s(z)} \, d\signal{s} \right] \, dz \nonumber\\
    & = \int_{0}^{1}\allocation(\signal{s}) \left(v_b({\signal{s}}) - v_s(\signal{s}) \right) \, d\signal{s} - \int_{0}^1{\allocation(z) v'_s(z)} \, dz + \int_{0}^{1} {(1-z)\allocation(z) v'_s(z)} \, dz \nonumber\\
    & = \int_{0}^{1}\allocation(\signal{s}) \left(v_b({\signal{s}}) - v_s(\signal{s}) - \signal{s}v'_s(s) \right) \, d\signal{s} \tag*{Substituting variable \text{$z$} with \text{$\signal{s}$}} \nonumber\\
    & = \int_{0}^{1}\allocation(\signal{s}) \left( k \signal{s}^k - \signal{s}^k - \signal{s} \cdot k \signal{s}^{k-1}\right)\, d\signal{s} \tag*{Plugging in \text{$v_b(\signal{s})$ and $v_s(\signal{s})$}} \nonumber \\
    & = - \int_{0}^{1}\allocation(\signal{s}) \signal{s}^k \, d\signal{s}. \label{ineq:(1,0)-expected-buyer-utility-bound}
\end{align}

    Notice that if $\allocation(\signal{s}) \neq  0$ for any $\signal{s}$, then our bound from ~\eqref{ineq:(1,0)-expected-buyer-utility-bound} implies that the buyer's expected utility will become negative, which violates the buyer's interim IR constraint. As such, we can argue that the only mechanism that satisfies BIC and interim IR for both the buyer and the seller is the mechanism that never trades, i.e., $\allocation(\signal{s}) = \payment(\signal{s}) = 0$ for all $\signal{s} \in [0,1]$. Finally, since $v_b(\signal{s}) \geq v_s(\signal{s})$ for all $\signal{s} \in [0,1]$, the approximation ratio can be computed as:
    \[
    \frac{OPT}{ALG} = \frac{\mathbb{E}_{\signal{s}} [v_b(\signal{s})]}{\mathbb{E}_{\signal{s}} [v_s(\signal{s})]} = k.
    \] 
    
    % \kgnote{Need to properly define the allocation/payments as the constant 0 functions, may make more sense once we've properly defined mechanisms.}

    \end{proof}

%% file: obstacles.tex
\section{Obstacles to Designing Mechanisms for Bilateral Trade with Interdependent Values}

\subsection{The Necessity of Bayesian IC and Interim IR in Interdependent Bilateral Trade}\label{subsec-obstacle-ex-post}
In this section, we motivate our incentive concepts: Bayesian incentive-compatibility as our truthfulness notion and interim individual rationality.  We explain why stronger notions like dominant-strategy, ex-post incentive compatibility, or ex-post individual rationality cannot be meaningfully applied. First, we define these other notions.

\begin{definition}[Ex-Post Incentive Compatibility (EPIC)]\label{def:EPIC}
    A direct mechanism $(\allocation, \payment)$ is 
{\em ex-post incentive-compatible} if for every $\signal{b}$ and $\signal{s} \in [0,1]$, 
\begin{equation*}
\allocation(\signal{b},\signal{s}) \cdot v_{b}(\signal{b},\signal{s}) - \payment(\signal{b},\signal{s}) \geq \allocation(\signal{b}',\signal{s}) \cdot v_{b}(\signal{b},\signal{s}) - \payment(\signal{b}',\signal{s}),
\label{ineq:EPIC-buyer}
\end{equation*}
and
\begin{equation*}
\payment(\signal{b},\signal{s}) - \allocation(\signal{b},\signal{s}) \cdot v_{s}(\signal{b},\signal{s}) \geq \payment(\signal{b},\signal{s}') - \allocation(\signal{b},\signal{s}') \cdot v_{s}(\signal{b},\signal{s}).
\label{ineq:EPIC-seller}
\end{equation*}
\end{definition}

\begin{definition}[Ex-Post Individually Rational (Ex-post IR)]\label{def:EPIR}
    A direct mechanism $(\allocation, \payment)$ is 
{\em ex-post individually rational} for every $\signal{b}$ and $\signal{s} \in [0,1]$, 
\begin{equation*}
\allocation(\signal{b},\signal{s}) \cdot v_{b}(\signal{b},\signal{s}) - \payment(\signal{b},\signal{s})\geq 0, \hspace{1cm} \text{and} \hspace{1cm} \payment(\signal{b},\signal{s}) - \allocation(\signal{b},\signal{s}) \cdot v_{s}(\signal{b},\signal{s}) \geq 0.
\label{ineq:EPIR}
\end{equation*}
\end{definition}

We show that in bilateral trade with interdependent values, ex-post IC and ex-post IR may be overly restrictive, to the extent that a mechanism satisfying these constraints may no longer be able to approximate the optimal expected welfare. Of course, one can also consider stronger notions like dominant-strategy incentive-compatible (DSIC) mechanisms, but since every DSIC mechanism is also ex-post IC, our impossibility below applies to DSIC mechanisms as well.

\begin{claim}\label{claim:EPIC-EPIR-impossibility}
    For every $c > 2$, there exists an information structure where no ex-post IC and IR mechanism achieves a constant approximation ratio better than $c$, while there exists a BIC and interim IR mechanism that achieves almost optimal social welfare.
\end{claim}

\begin{proof}
    Let the seller's signal $\signal{s}$ be drawn uniformly from $[0,1]$. Fix some $\epsilon>0$ and let $c > 2 + 2 \epsilon$. Consider the following valuations:
    $$
        v_b(\signal{s})  =  c\signal{s} \hspace{2cm} \text{and} \hspace{2cm} 
        v_s(\signal{s})  = \signal{s} + \epsilon.
    $$
    % We aim to investigate what happens with small $\epsilon > 0$To this end 
    Consider now any mechanism $\mechanism = (\allocation,\payment)$ that satisfies ex-post IC and ex-post IR. 
    % At signal $\signal{s}=0$, the buyer valuation is $v_b(0) = 0$ and the seller valuation is $v_s(0) = \epsilon$. 
    We start by considering the ex-post IR conditions at $\signal{s} = 0$:
    \begin{align*}
        \allocation(0) \cdot v_b(0) - \payment(0) \geq 0 \Rightarrow \payment(0) = 0, \\
        \payment(0) - \allocation(0) \cdot v_s(0) \geq 0 \Rightarrow \allocation(0) = 0.
    \end{align*}

    We now argue that in fact $\allocation(\signal{s}) = 0$ and $\payment(\signal{s}) = 0$ for all $\signal{s} \in [0,1]$. This is because the functions $\allocation$, $\payment$ are non-increasing in $\signal{s}$ (which is due to the ex-post IC condition and the non-decreasing monotonicity of $v_s$), which combined with $\allocation(0) = 0$ and $\payment(0) = 0$, implies that $\allocation(\signal{s}) = 0$ and $\payment(\signal{s})=0$ for all $\signal{s}\in[0,1]$.  Hence the no-trade no-payment mechanism is the only mechanism that satisfies ex-post IC and ex-post IR. We now compute the approximation ratio of that mechanism:
    \[
    \frac{OPT}{ALG} = \frac{\E_{\signal{s}}[\max(v_s(\signal{s}),v_b(\signal{s}))]}{\E_{\signal{s}}[v_s(\signal{s})]} \geq \frac{\E_{\signal{s}}[v_b(\signal{s})]}{\E_{\signal{s}}[v_s(\signal{s})]} =  \frac{c\cdot \E_{\signal{s}}[\signal{s}]}{\epsilon + \E_{\signal{s}}[\signal{s}]} = \frac{c}{2 \epsilon + 1} \approx c.
    \] 
    
    We now provide a BIC and interim IR mechanism that almost achieves the optimal expected social welfare. The mechanism posts a price of $q=1 + \epsilon$. Notice that in this mechanism, the dominant strategy of the seller is to always accept the trade (since $v_s(\signal{s}) \leq q$ for all $\signal{s}\in[0,1]$), and this reveals no information about their signal to the buyer. 
    %On the other hand, the buyer's only options are to either accept or reject the trade. 
    
    %The proposed mechanism is Bayesian incentive-compatible for the seller since their response is irrespective of their signal and they are not revealing anything to the mechanism. 
    The mechanism is also BIC for the buyer simply because the buyer is fully uninformed and the seller always accepts regardless of their signal. We only need to check the buyer's interim IR condition. If they select to always reject the trade, then their expected utility is trivially 0. If they select to accept the trade then their expected utility is:
    \[
    \mathbb{E}_{\signal{s}}[v_b(\signal{s}) - q] = \frac{c}{2} - 1 - \epsilon = \frac{c-2-2\epsilon}{2} > 0.
    \]
    As a result, the buyer prefers to always accept the trade, and the mechanism satisfies buyer IR as well. Note that the optimal welfare is at most the $E[v_b]+E[v_s]=\frac {c}2 + \frac {1+\epsilon} 2$. Our posted price mechanism always trades the item, so its welfare is $\frac c 2$. Thus, its approximation ratio approaches $1$ as $c$ goes to infinity, and $\epsilon$ to $0$.
    %\ttnote{We can get rid of this computations and simply intuitively argue that $OPT \approx \mathbb{E}_{\signal{s}}[v_b(\signal{s})]$}
    %\begin{align*}\label{eq:EPIC-impos-OPT}
    %OPT & = \E_{\signal{s}}[\max(v_s(\signal{s}),v_b(\signal{s}))] \nonumber\\
    %& = \E_{\signal{s}}[v_s(\signal{s})) | v_b(\signal{s})) \leq v_s(\signal{s}))] \mathbb{P}[v_b(\signal{s})) \leq v_s(\signal{s}))] + %\E_{\signal{s}}[v_b(\signal{s})) | v_b(\signal{s})) > v_s(\signal{s}))] \mathbb{P}[v_b(\signal{s})) > v_s(\signal{s}))] \nonumber \\
    %& = \int_0^{\frac{\epsilon}{c-1}} \left(\signal{s} + \epsilon \right) d\signal{s} + \int_{\frac{\epsilon}{c-1}}^1 \left(c \cdot %\signal{s}\right) d\signal{s} \nonumber \\
    %&=\frac{\epsilon^2}{2(c-1)^2}+\frac{\epsilon^2}{c-1}+\frac{c}{2}\left(1-\frac{\epsilon^2}{(c-1)^2}\right) \nonumber \\
    %&=\frac{c}{2} + \frac{\epsilon^2}{2(c-1)}
    %% = \frac{c^2 -2c + \epsilon^2}{2(c-1)}
    %\end{align*}    
    %Note that the approximation ratio of this mechanism is:
    %\[
    %\frac{OPT}{ALG} = \frac{\frac{c}{2} + \frac{\epsilon^2}{2(c-1)}}{\frac{c}{2}} =  1 + \frac{\epsilon^2}{2c(c-1)} \approx 1
    %\]
\end{proof}

% \ttedit{To further motivate how allowing for BIC mechanisms resolves some of these restrictions, we mention that the buyer offering mechanism (which is BIC, interim IR for the buyer) in this example would set a price of $1+\epsilon$, which the seller always agree to. As a result the approximation ratio of this mechanism would be roughly 1.}

%\paragraph{On different notions of IC and IR for each side of the market.} We remark that the restrictions from ex-post IC and IR become prohibitive when we require that the mechanism satisfies these conditions for both players \emph{simultaneously}. This can already be observed on the instance we just studied since the proposed BIC and interim IR mechanism is actually DSIC and ex-post IR for the seller. Additionally, as we can see in our results in Section \ref{sec:seller-informed}, there exist mechanisms that approximate social welfare in a general information structure and satisfy ex-post IC (in fact, DSIC) and ex-post IR for the seller and BIC and interim IR for the buyer. 

\subsection{A Posted Price Mechanism which Never Trades}\label{subsec-obstacle-posted-price}

We show an example of simple linear valuation functions where no posted price mechanism admits an equilibrium in which the item is traded with positive probability.

\begin{example}\label{example:no-non-trivial-equilibrium}
   Let the valuations be as follows:
    \begin{align*}
            v_b(\signal{b}, \signal{s}) = \signal{b} + 2\cdot\signal{s}\\
            v_s(\signal{b}, \signal{s}) = \signal{s}+ 2\cdot\signal{b}
    \end{align*}
\end{example}

Consider a posted price mechanism $p$ for the information structure of Example~\ref{example:no-non-trivial-equilibrium}.
As before, the equilibrium strategies of the players are threshold strategies, as in Lemma \ref{lemma:equ-strategies-thresholds}. If an equilibrium strategy exists where both thresholds lie in $[0,1]$, then the following two inequalities must hold. The first ensures that 
the buyer's profit is non-negative for $\signal{b} = \thsbuyer$, and the second ensures that the seller's profit is non-negative for $\signal{s} = \thsseller$.

\begin{align}
    \thsbuyer + 2\cdot\dfrac{\thsseller}{2} & \geq  p \iff \thsbuyer + \thsseller \geq  p \nonumber\\
    \thsseller + 2\cdot\frac{1+\thsbuyer}{2} & \leq  p \iff \thsseller + 1+\thsbuyer \leq  p \nonumber
\end{align}

However, these two inequalities have no solution. Therefore, for every price $p$, the only equilibrium strategy is to never trade the item.

%% file: appproofs.tex
\section{Proof of Lemma~\ref{lemma:buyer-signal-tec-summation}}\label{sec:app:lemma:aux-lemma}

\begin{proof}[Proof of Lemma~\ref{lemma:buyer-signal-tec-summation}] We bound a term arising from the payment identity, showing that the sum of allocation probabilities has to be smaller than a certain quantity.
\begin{align}
        \sum_{\signal{b}=1}^m \frac{1}{k^{\signal{b}}}&\cdot\int_{1}^{\signal{b}} \allocation(b,\signal{s})\cdot k^b\cdot\ln{k}\d b  = \sum_{\signal{b}=1}^m \frac{1}{k^{\signal{b}}}\cdot \sum_{i =1}^{\signal{b}-1}  \int_{i}^{i+1} \allocation(b,\signal{s})\cdot k^b\cdot\ln{k} \d b \nonumber\\
        & = \sum_{i=1}^{m-1} \left(\int_{i}^{i+1} \allocation(b,\signal{s})\cdot k^b\cdot\ln{k} \d b \right)\cdot \sum_{\signal{b}=i+1}^m \frac{1}{k^{\signal{b}}} \nonumber & \text{order of integration} \\
        & =  \sum_{i=1}^{m-1} \left(\int_{i}^{i+1} \allocation(b, \signal{s})\cdot k^b\cdot\ln{k} \d b \right)\cdot \left(\frac{\frac{1}{k^{i+1}}\left(\frac{1}{k^{m-i}}-1\right)}{\frac{1}{k} -1}\right) \nonumber & \text{geometric sum} \\
        & = \sum_{i=1}^{m-1} \left(\int_{i}^{i+1} \allocation(b,\signal{s})\cdot k^b\cdot\ln{k}\d b \right)\cdot \left(\frac{1}{k^{i}(k-1)} -\frac{1}{k^m\cdot (k-1)}\right) \nonumber \\
        & \geq \sum_{\signal{b}=1}^{m-1} \left( k^{\signal{b}+1} -k^{\signal{b}} \right)\cdot \allocation(\signal{b}, \signal{s})\cdot \left(\frac{1}{k^{\signal{b}}(k-1)} -\frac{1}{k^m\cdot (k-1)}\right) \nonumber & \text{monotonicity of $\allocation(\cdot)$} \\
        & = \sum_{\signal{b}=1}^{m-1} \allocation(\signal{b},\signal{s})\left(1-\frac{1}{k^{m-\signal{b}}}\right). \nonumber
\end{align}
\end{proof}

\section{Proof of Theorem \ref{thm:uninformed-buyer-impossibility} Part I}
%\section{Missing Proofs from Section~\ref{sec:uninformed-buyer}}
\label{sec:app-uninformed-buyer}

In this section, we provide the proof of the first part of Theorem \ref{thm:uninformed-buyer-impossibility}.  For convenience, we restate the theorem. 

%{numbredtheorem}

\begin{numbredtheorem}{\ref{thm:uninformed-buyer-impossibility} Part I}%The following two statements hold:
%\begin{enumerate}
%\item 
For every $k, m \in \mathbb{N}$, there exists an information structure where the seller is $\informedcoeffseller$-informed and the buyer is $\informedcoeffbuyer$-informed such that $\informedcoeffseller > 0, \informedcoeffbuyer < 1, k\geq 4, m\geq 5$, and no BIC and interim IR mechanism can have an approximation ratio better than:
    \begin{equation*}
\frac{\informceoffbuyer+1}{\frac{1}{k}(1+ \informceoff)+ \informceoffbuyer  +\frac{4}{m}+\informceoffbuyer\cdot \frac{2}{k}+\frac{2}{k^2}\cdot\informceoff},
    \end{equation*}
    where $\informceoffbuyer =\frac{\informedcoeffbuyer}{1-\informedcoeffbuyer}$ and $ \sellerratioinformed = \frac{1-\informedcoeffseller}{\informedcoeffseller}$.
%\item As a corollary. for every $\informedcoeffseller > 0$, and every $c>1$, there exists an information structure where the buyer is totally uninformed and the seller is $\informedcoeffseller$-informed, and no BIC and interim IR mechanism provides an approximation ratio better than $c$.
%\end{enumerate}
\end{numbredtheorem}

%%%%%% PROOF OF THEOREM PART I

%\subsection{Proof of Part I of Theorem \ref{thm:uninformed-buyer-impossibility}}
%\label{sec:construction:seller:has:signal}

\begin{proof}%[Proof of Theorem \ref{thm:uninformed-buyer-impossibility} Part I]

We construct an information structure $\mathcal{I}_{k,m,\informedcoeffseller, \informedcoeffbuyer}$ which is similar in spirit to the one constructed in Section~\ref{sec:construction:buyer:has:signal}. However, now the informative signal is held by the seller rather than the buyer, which makes the technical details of the proof different.

In the information structure $\mathcal{I}_{k, m, \informedcoeffseller, \informedcoeffbuyer}$, the seller has an informative signal $\signal{s}$. This signal is a discrete random variable with support $[m]\cup \{0\}$. The buyer has a less informative signal $\signal{b}$.  The seller's valuation function is defined as $v_s(\signal{b}, \signal{s}) = \sellervaluationbuyerfunc{}(\signal{b}) + \sellervaluationsellerfunc{}(\signal{s})$. The buyer's valuation is a function of the seller's signal plus a constant: $v_b(\signal{b}, \signal{s}) = \buyervaluationbuyerfunc{}(\signal{b}) + \buyervaluationsellerfunc{}(\signal{s})$, where $\buyervaluationsellerfunc{}(\signal{s})$, is a constant function $\delta \cdot \mu$, with $\mu = \frac{m}{k}, \, \delta=\frac{\informedcoeffbuyer}{1-\informedcoeffbuyer}, \, \informceoff=\frac{\informedcoeffseller}{1-\informedcoeffseller}$, and $\mu = \E_{\signal{b}}[\buyervaluationbuyerfunc{}(\signal{b})]$. 

We now specify the probability mass function of the signal $\signal{s}$, the buyer's valuation function, and the seller's valuation function. 
        \begin{equation*}\label{equ:def-prob-and-val-uninformed-buyer-impossibility}
    \pdfsellersignal(\signal{s}) = \begin{cases}
                 \frac{1}{k^{\signal{s}}} &  \signal{s} \in [m];\\
        1-\sum_{i=1}^m\frac{1}{k^i}  & \signal{s}=0.
           \end{cases} \quad
    v_b(\signal{s}) = \begin{cases}
                 k^{\signal{s}} +\informceoffbuyer\cdot\mu  &   \signal{s} \in [m] ;\\
                 \informceoffbuyer\cdot\mu  &  \signal{s} \notin [m].
           \end{cases} \quad 
    v_s(\signal{b}, \signal{s}) = \begin{cases}
                 k^{\signal{s}-1} + \sellervaluationbuyerfunc{}(\signal{b})  &   \signal{s} \in [m] ;\\
                 \sellervaluationbuyerfunc{}(\signal{b}) &  \signal{s} \notin [m].
           \end{cases}           
\end{equation*}

As in Section~\ref{sec:construction:buyer:has:signal}, the informative signal in $\mathcal{I}_{k,m, \informedcoeffseller, \informedcoeffbuyer}$ is not distributed uniformly over $[0,1]$ for presentation, but a simple adjustment fixes this.
%This description simplifies the presentation, but the information structure can easily be adjusted to ensure a uniform distribution of signals.
Now, let $\mechanism= (\allocation, \payment)$ be a BIC and interim IR mechanism for this information structure. 
We first state Lemma~\ref{lemma:uninformed-buyer-interim-ir-bound-on-alloc}. We then prove the first part of the theorem by applying the lemma. The proof of Lemma~\ref{lemma:uninformed-buyer-interim-ir-bound-on-alloc} is deferred to Appendix~\ref{sec:app-uninformed-buyer}.

\begin{lemma}\label{lemma:uninformed-buyer-interim-ir-bound-on-alloc}
    If $\mechanism$ is interim IR, then:
    $$
    \informceoffbuyer\cdot\mu\cdot\frac{2}{k} + \E_{\signal{b}}[\sellervaluationbuyerfunc{}(\signal{b})]\cdot\frac{2}{k} + 4  \geq \sum_{\signal{s}=1}^m \E_{\signal{b}}[\allocation(\signal{b},\signal{s})].
    $$
\end{lemma}

We start by computing the expected value of both agents:
    \begin{equation*}
    % \label{equality:impossibility-buyer-expected-value}
        \E[v_b(\signal{s})] = \informceoffbuyer\cdot \mu +\sum_{\signal{s}=1}^m \frac{1}{k^{\signal{s}}}\cdot k^{\signal{s}} = m(\informceoffbuyer + 1) \hspace{1cm} \text{and} \hspace{1cm}  \E[v_s(\signal{b}, \signal{s})] = \E_{\signal{b}}[\sellervaluationbuyerfunc{}(\signal{b})] + \sum_{\signal{s}=1}^m \frac{1}{k^{\signal{s}}}\cdot k^{\signal{s}-1} = \frac{m}{k}(1+ \informceoff).
    \end{equation*}
    The expected welfare of $\mechanism$'s is:
    \begin{align*}
        ALG & = \E[v_s(\signal{b},\signal{s})] + \E_{\signal{b}}[\sum_{\signal{s}=0}^m \pdfsellersignal(\signal{s})\cdot\allocation(\signal{b},\signal{s})\left( v_b(\signal{s}) -v_s(\signal{b},\signal{s}) \right)]\\
        % & =\frac{m}{k}\cdot\frac{1}{\informedcoeffseller} - \pdfsellersignal(0)\cdot\allocation(0)\cdot\sellerratioinformed\cdot\sellerexpectedfrombuyersignal + \sum_{\signal{s}=1}^m \frac{\allocation(\signal{s})}{k^{\signal{s}}}\cdot\left(k^{\signal{s}} -k^{\signal{s}-1}- \sellerratioinformed\cdot\sellerexpectedfrombuyersignal\right) \\
        & \leq \frac{m}{k}(1+ \informceoff)+ \informceoffbuyer\cdot\ m +(1-\frac{1}{k})\cdot\sum_{\signal{s}=1}^m \E_{\signal{b}}[\allocation(\signal{b},\signal{s})]\\
        & \leq \frac{m}{k}(1+ \informceoff)+ \informceoffbuyer\cdot\ m +\sum_{\signal{s}=1}^m \E_{\signal{b}}[\allocation(\signal{b},\signal{s})].
    \end{align*}

Note that the optimal welfare is at least the buyer's expected value. Thus, using Lemma~\ref{lemma:uninformed-buyer-interim-ir-bound-on-alloc} the approximation ratio of $\mechanism$ is at least:

    \begin{equation*}
        \frac{OPT}{ALG} \geq \frac{m(\informceoffbuyer+1)}{\frac{m}{k}(1+ \informceoff)+ \informceoffbuyer\cdot m  +4+\informceoffbuyer\cdot m\frac{2}{k}+\frac{2}{k}\frac{m}{k}\cdot\informceoff} = \frac{\informceoffbuyer+1}{\frac{1}{k}(1+ \informceoff)+ \informceoffbuyer  +\frac{4}{m}+\informceoffbuyer\cdot \frac{2}{k}+\frac{2}{k^2}\cdot\informceoff}.
        % = \frac{m}{(\sellerratioinformed+1 )\cdot \frac{m}{k} +\sum_{\signal{s}=1}^m \allocation(\signal{s})(1-\frac{1}{k})}
    \end{equation*}

    \end{proof}

%%%%%% LEMMAS

Before proving Lemma~\ref{lemma:uninformed-buyer-interim-ir-bound-on-alloc}, we state a technical lemma we use for the proof.

\begin{lemma}\label{lemma:seller-has-signal-tec}
For every $\signal{b}$:
    $$
    \sum_{\signal{s}=1}^m \frac{1}{k^{\signal{s}}}\cdot \left(\int_0^1 \allocation(\signal{b},t) \d t +\int_{1}^{\signal{s}} \allocation(\signal{b}, t)\cdot k^{t-1}\ln{k}\d t \right ) \leq \frac{1}{k-1}\cdot \allocation(\signal{b},0) + \sum_{\signal{s}=1}^{m-1} \allocation(\signal{b},\signal{s})\left(\frac{1}{k}-\frac{1}{k^{m-\signal{s}+1}}\right). 
    $$
\end{lemma}

We now continue with the proof of Lemma~\ref{lemma:uninformed-buyer-interim-ir-bound-on-alloc}.

\begin{proof}[Proof of Lemma~\ref{lemma:uninformed-buyer-interim-ir-bound-on-alloc}]By Myerson's Lemma for the seller (Lemma~\ref{lemma:seller-payment-formula}), the payment function satisfies Equation~\eqref{equ:payment-formula-seller-uninformed-buyer-impossibility}, for every $\signal{s} \in [m]$. Equation ~(\ref{ineq:uninformed-buyer-impossibility-bound-on-p}) is by $\mechanism$'s interim IR for the seller when $\signal{s}=0$.

% \begin{equation}\label{equ:payment-formula-seller-uninformed-buyer-impossibility}
%     \payment(\signal{s}) = \payment(0) + \allocation(\signal{s})(k^{\signal{s}-1} +\sellerratioinformed\cdot\sellerexpectedfrombuyersignal) -\allocation(0)\cdot\sellerratioinformed\cdot\sellerexpectedfrombuyersignal-\int_0^{k^{\signal{s}}-1} \allocation(t) \d t 
% \end{equation}

% \ttnote{Rearrange Myerson payment to have "Util(s) = Util(0) - integrals" and invoke Interim IR at $\signal{s}$. Also expand $v_s(\signal{b},0)$ from here, because it seems a bit out of the blue later on. Seems more cohesive like that.}
\begin{align}
    \E_{\signal{b}} \left[ \payment(\signal{b},\signal{s}) \right] 
    &= \E_{\signal{b}} \bigg[
        \payment(\signal{b},0) 
        - v_s(\signal{b},0) \cdot \allocation(\signal{b},0) 
        + v_s(\signal{b},\signal{s}) \cdot \allocation(\signal{b},\signal{s})  \notag \\
    &\quad \quad - \int_{0}^{\signal{s}} \allocation(\signal{b},t) 
        \left. \frac{\partial v_s(\signal{b},\signal{s})}{\partial \signal{s}} \right|_{\signal{s}=t} \, d t 
    \bigg] \notag \\
    &= \E_{\signal{b}} \bigg[
        \payment(\signal{b},0) 
        - v_s(\signal{b},0) \cdot \allocation(\signal{b},0) 
        + v_s(\signal{b},\signal{s}) \cdot \allocation(\signal{b},\signal{s}) \notag \\
    &\quad \quad - \int_0^1 \allocation(\signal{b},t)\, d t  
        - \int_1^{\signal{s}} \allocation(\signal{b},t) \cdot k^{t-1} \ln{k} \, d t
    \bigg].
    \label{equ:payment-formula-seller-uninformed-buyer-impossibility}
\end{align}

\begin{align}\label{ineq:uninformed-buyer-impossibility-bound-on-p}
    \E_{\signal{b}} \left[ \payment(\signal{b},\signal{s}) \right] & \geq \E_{\signal{b}} \left[ v_s(\signal{b},\signal{s}) \cdot \allocation(\signal{b},\signal{s})-\int_0^1 \allocation(\signal{b},t)\d t  - \int_1^{\signal{s}} \allocation(\signal{b},t)\cdot k^{t-1}\cdot\ln{k} \d t\right]. 
\end{align}

Since $\mechanism$ is interim IR for the seller, we obtain the following for every $\signal{s} \in[m]$:

\begin{align*}
    &\E_{\signal{b}} 
    \left[-v_s(\signal{b}, \signal{s}) \cdot \allocation(\signal{b}, \signal{s}) + \payment(\signal{b}, \signal{s}) \right] \geq 0 \\[8pt]
    \iff \quad & \E_{\signal{b}} \bigg[
        \payment(\signal{b}, 0) 
        - \allocation(\signal{b},0) \cdot \sellervaluationbuyerfunc{}(\signal{b}) 
        - \int_0^1 \allocation(\signal{b},t)\, d t  
        - \int_1^{\signal{s}} \allocation(\signal{b},t) \cdot k^{t-1} \ln{k} \, d t
    \bigg] \geq 0 \\[8pt]
    \iff \quad & \E_{\signal{b}}[\payment(\signal{b},0)] \geq  
    \E_{\signal{b}} \bigg[
        \allocation(\signal{b},0) \cdot \sellervaluationbuyerfunc{}(\signal{b}) 
        + \int_0^1 \allocation(\signal{b},t) \, d t  
        + \int_1^{\signal{s}} \allocation(\signal{b},t) \cdot k^{t-1} \ln{k} \, d t
    \bigg].
\end{align*}

As $\allocation(\signal{b},\signal{s})$ is monotone non-increasing in $\signal{s}$, applying the previous inequality at $\signal{s}=m$ provides the following bound on $\payment(\signal{b},0)$:
\begin{align}\label{ineq:uninformed-buyer-bound-on-p-zero}
    \E_{\signal{b}}[\payment(\signal{b},0)] 
    % &\geq \E_{\signal{b}}\left[\allocation(\signal{b},0)\cdot\sellervaluationbuyerfunc{}(\signal{b}) 
    % + \int_0^1 \allocation(\signal{b},t) \d t 
    % + \int_1^{m} \allocation(\signal{b},t)\cdot k^{t-1}\ln{k} \d t\right] \notag \\
    &    \geq \E_{\signal{b}}\left[\allocation(\signal{b},0)\cdot\sellervaluationbuyerfunc{}(\signal{b}) 
    + \sum_{\signal{s}=1}^{m-1} \int_{\signal{s}}^{\signal{s}+1} \allocation(\signal{b},t)\cdot k^{t-1}\ln{k}\d t \right] \notag \\
    &\geq \E_{\signal{b}}\left[\allocation(\signal{b},0)\cdot\sellervaluationbuyerfunc{}(\signal{b}) 
    + \sum_{\signal{s}=1}^{m-1} \allocation(\signal{b},\signal{s}+1)\cdot k^{\signal{s}-1}(k-1)\right].    
\end{align}

Next, we use the interim IR condition for the buyer to conclude the proof of the lemma.

% \ttnote{Plug in $v_b(\signal{b},0)$, since you are doing the same for $v_b(\signal{b},\signal{s})$ in the sum}
\begin{align}\label{ineq:uninformed-buyer-impos-interim-buyer}
    &  \E_{\signal{b}}\left[\pdfsellersignal(0)\cdot\left(\allocation(\signal{b},0)\cdot v_b(\signal{b},0) -\payment(\signal{b},0)\right) + \sum_{\signal{s}=1}^m \frac{1}{k^\signal{s}}\left((k^{\signal{s}}+ \informceoffbuyer\cdot \mu)\cdot\allocation(\signal{b},\signal{s}) -\payment(\signal{b}, \signal{s})\right)\right] \geq 0 .
    % &  \E_{\signal{b}}\left[-\pdfsellersignal(0)\cdot\payment(\signal{b},0) + \sum_{\signal{s}=1}^m \allocation(\signal{b},\signal{s}) -\frac{\payment(\signal{b},\signal{s})}{k^{\signal{s}}}\right] \geq 0 
\end{align}

% \ttnote{On the second inequality I don't see why its the sign of $\E_{\signal{b}} [\sellervaluationbuyerfunc{}(\signal{b})]$ is $\boldsymbol{+}$ and not $\boldsymbol{-}$.}\asnote{because in the payment formula we have $-v_s(\signal{b}, 0)\cdot\allocation(\signal{b}, 0)$ and in the inequality we have $-p$}
We now analyze the left hand side of Inequality~\eqref{ineq:uninformed-buyer-impos-interim-buyer}.  
\begin{align}
    &\E_{\signal{b}} \bigg[
        \pdfsellersignal(0) \cdot \left(
            \allocation(\signal{b},0) \cdot v_b(\signal{b},0) 
            - \payment(\signal{b},0) 
        \right) 
        + \sum_{\signal{s}=1}^m \frac{1}{k^{\signal{s}}} \left(
            (k^{\signal{s}} + \informceoffbuyer \cdot \mu) \cdot \allocation(\signal{b},\signal{s}) 
            - \payment(\signal{b}, \signal{s}) 
        \right) 
    \bigg] \notag \\
    &\leq \informceoffbuyer \cdot \mu 
    + \E_{\signal{b}} \bigg[
        -\pdfsellersignal(0) \cdot \payment(\signal{b},0) 
        + \sum_{\signal{s}=1}^m \allocation(\signal{b},\signal{s}) 
        - \frac{\payment(\signal{b},\signal{s})}{k^{\signal{s}}}
    \bigg] \notag \\
    &\leq \informceoffbuyer \cdot \mu 
    + \E_{\signal{b}} [\sellervaluationbuyerfunc{}(\signal{b})] 
    + \E_{\signal{b}} \bigg[
        -\payment(\signal{b},0) 
        + \sum_{\signal{s}=1}^m \allocation(\signal{b},\signal{s}) \left(1 - \frac{1}{k}\right) \notag \\
    &\quad\quad
        + \sum_{\signal{s}=1}^m \frac{
            \int_0^1 \allocation(\signal{b},t) \, d t 
            + \int_1^{\signal{s}} \allocation(\signal{b},t) \cdot k^{t-1} \ln{k} \, d t
        }{k^{\signal{s}}}
    \bigg].
    \label{ineq:seller-has-signal-first-half}
\end{align}

The second Inequality is by using Equality~\eqref{equ:payment-formula-seller-uninformed-buyer-impossibility} and a bound of $1$ on the allocation function.
We now continue to bound some of this expression:
\begin{align*}
    &  \E_{\signal{b}}\left[-\payment(\signal{b},0) +\sum_{\signal{s}=1}^m \allocation(\signal{b},\signal{s})\left(1-\frac{1}{k}\right) +  \sum_{\signal{s}=1}^m \frac{\int_0^1 \allocation(\signal{b},t)\d t + \int_1^{\signal{s}} \allocation(\signal{b},t)\cdot k^{t-1}\cdot\ln{k} \d t }{k^{\signal{s}}}\right]\\
     & \leq \E_{\signal{b}}\left[-\left(\sum_{\signal{s}=1}^{m-1} \allocation(\signal{b},\signal{s}+1)\cdot k^{\signal{s}-1}(k-1)\right) + \sum_{\signal{s}=1}^m \allocation(\signal{b},\signal{s})(1-\frac{1}{k}) + \sum_{\signal{s}=1}^m \frac{\int_0^1 \allocation(\signal{b},t)\d t + \int_1^{\signal{s}} \allocation(\signal{b},t)\cdot k^{t-1}\cdot\ln{k} \d t }{k^{\signal{s}}}\right]\\
     % & \leq \E_{\signal{b}}\left[-\sum_{\signal{s}=1}^{m-1} \allocation(\signal{b},\signal{s}+1)\cdot k^{\signal{s}-1}(k-1) + \sum_{\signal{s}=1}^m \allocation(\signal{b},\signal{s})(1-\frac{1}{k}) + \sum_{\signal{s}=1}^m \frac{\int_0^1 \allocation(t)\d t \int_1^{\signal{s}} \allocation(t)\cdot k^{t-k}\cdot\ln{k} \d t }{k^{\signal{s}}}\right]\\
     & \leq \E_{\signal{b}}\left[-\sum_{\signal{s}=1}^{m-1} \allocation(\signal{b},\signal{s}+1)\cdot k^{\signal{s}-1}(k-1) + \sum_{\signal{s}=1}^m \allocation(\signal{b},\signal{s})(1-\frac{1}{k}) + \frac{1}{k-1}\cdot \allocation(\signal{b},0) + \sum_{\signal{s}=1}^{m-1} \allocation(\signal{b},\signal{s})\left(\frac{1}{k}-\frac{1}{k^{m-\signal{s}+1}}\right)\right]\\
     & \leq \E_{\signal{b}} \left[\sum_{\signal{s}=2}^{m-1} \allocation(\signal{b},\signal{s})\left(1-\frac{1}{k^{m-\signal{s}+1}}-k^{\signal{s}-2}(k-1)\right) +\allocation(\signal{b},0)\frac{1}{k-1} + \allocation(\signal{b},1)\cdot\left(1-\frac{1}{k^m}\right)\right]\\
     & \leq \E_{\signal{b}}\left[2k + (2-k)\cdot\sum_{\signal{s} =1}^{m}\allocation(\signal{b},\signal{s})\right]\leq \E_{\signal{b}}\left[2k -\frac{k}{2}\cdot\sum_{\signal{s}=1}^m \allocation(\signal{b},\signal{s})\right] = 2k -\frac{k}{2}\sum_{\signal{s}=1}^m \E_{\signal{b}}[\allocation(\signal{b},\signal{s})].
     % \asnote{\leq 3k -\frac{k}{2}\sum_{\signal{s}=0}^m \E_{\signal{b}}[\allocation(\signal{b},\signal{s})]}
\end{align*}

The first inequality is by Inequality~\eqref{ineq:uninformed-buyer-bound-on-p-zero}, and the second is by Lemma~\ref{lemma:seller-has-signal-tec}. The forth inequality is due to $k^{\signal{s}-2} \geq 1$ for $\signal{s} \geq 2$.
Now, if $\mechanism$ satisfies interim IR for the buyer, then it must be that 
$$
 \informceoffbuyer\cdot\mu + \E_{\signal{b}}[\sellervaluationbuyerfunc{}(\signal{b})] + 2k -\frac{k}{2}\sum_{\signal{s}=1}^m \E_{\signal{b}}[\allocation(\signal{b},\signal{s})] \geq 0 \iff  \informceoffbuyer\cdot\mu\cdot\frac{2}{k} + \E_{\signal{b}}[\sellervaluationbuyerfunc{}(\signal{b})]\cdot\frac{2}{k} + 4  \geq \sum_{\signal{s}=1}^m \E_{\signal{b}}[\allocation(\signal{b},\signal{s})].
$$
% \asnote{
% \begin{align*}
%     \informceoffbuyer\cdot\mu + \E_{\signal{b}}[\sellervaluationbuyerfunc{}(\signal{b})] + 3k -\frac{k}{2}\sum_{\signal{s}=0}^m \E_{\signal{b}}[\allocation(\signal{b},\signal{s})] \geq 0 \iff  \informceoffbuyer\cdot\mu\cdot\frac{2}{k} + \E_{\signal{b}}[\sellervaluationbuyerfunc{}(\signal{b})]\cdot\frac{2}{k} + 6  \geq \sum_{\signal{s}=0}^m \E_{\signal{b}}[\allocation(\signal{b},\signal{s})]
% \end{align*}}

\end{proof}

\begin{proof}[Proof of Lemma~\ref{lemma:seller-has-signal-tec}]We bound a term arising from the payment identity, showing that the sum of allocation probabilities has to be smaller than a certain quantity.
    \begin{align}
    \label{equ:uninformed-buyersum-of-payments}
    \sum_{\signal{s}=1}^m & \frac{1}{k^{\signal{s}}}\cdot  \left(\int_0^1 \allocation(\signal{b},t) \d t +\int_{1}^{\signal{s}} \allocation(\signal{b}, t)\cdot k^{t-1}\ln{k}\d t \right ) \nonumber \\
        &= \sum_{\signal{s}=1}^m \frac{1}{k^{\signal{s}}}\cdot\int_0^1 \allocation(\signal{b},t) \d t + \sum_{\signal{s}=1}^m \frac{1}{k^{\signal{s}}}\cdot \sum_{i =1}^{\signal{s}-1}  \int_{i}^{i+1} \allocation(\signal{b},t)\cdot k^{t-1}\ln{k}\d t \nonumber\\
        & = \sum_{\signal{s}=1}^m \frac{1}{k^{\signal{s}}}\cdot\int_0^1 \allocation(\signal{b},t) \d t +\sum_{i=1}^{m-1} \left(\int_{i}^{i+1} \allocation(\signal{b},t)\cdot k^{t-1}\ln{k} \d t \right)\cdot \sum_{\signal{s}=i+1}^{m} \frac{1}{k^{\signal{s}}} \nonumber \\
        & =  \left(\frac{\frac{1}{k}\left(\frac{1}{k^m}-1\right)}{\frac{1}{k}-1}\right)\cdot\int_0^1 \allocation(\signal{b},t) \d t + \sum_{i=1}^{m-1} \left(\int_{i}^{i+1} \allocation(\signal{b},t)\cdot k^{t-1}\ln{k} \d t \right)\cdot \left(\frac{\frac{1}{k^{i+1}}\left(\frac{1}{k^{m-i}}-1\right)}{\frac{1}{k} -1}\right) \nonumber & \text{geometric sum} \\
        & =  \left(\frac{k^m -1}{k^m(k-1)}\right)\cdot\int_0^1 \allocation(\signal{b},t) \d t +\sum_{i=1}^{m-1} \left(\int_{i}^{i+1} \allocation(\signal{b},t)\cdot k^{t-1}\ln{k}\d t  \right)\cdot \left(\frac{1}{k^{i}(k-1)} -\frac{1}{k^m\cdot (k-1)}\right) \nonumber\\
        & \leq \frac{1}{k-1}\cdot \allocation(\signal{b},0) + \sum_{\signal{s}=1}^{m-1} \left( k^{\signal{s}} -k^{\signal{s}-1} \right)\cdot \allocation(\signal{b},\signal{s})\cdot \left(\frac{1}{k^{\signal{s}}(k-1)} -\frac{1}{k^m\cdot (k-1)}\right)\nonumber\\
        & = \frac{1}{k-1}\cdot \allocation(\signal{b},0) + \sum_{\signal{s}=1}^{m-1} \allocation(\signal{b},\signal{s})\left(\frac{1}{k}-\frac{1}{k^{m-\signal{s}+1}}\right).
        \end{align}
The last inequality is due to the monotonicity of $\allocation(\signal{s})$ (which is non-increasing in this case). 
\end{proof}